\documentclass[12pt]{amsart}
\usepackage{amssymb,amsthm,amsmath}
\usepackage[numbers,sort&compress]{natbib}
\usepackage{color}
\usepackage{graphicx}
\usepackage{tikz}
\usepackage{amssymb,amsthm,amsmath}
\usepackage[numbers,sort&compress]{natbib}
\usepackage{color}
\usepackage{graphicx}
\usepackage{tikz}
\usepackage{xparse}

\title[ ]{Quantitative  inductive estimates  for    Green's functions of   non-self-adjoint   matrices}

\author{Wencai Liu}
\address[W. Liu]{ Department of Mathematics, Texas A\&M University, College Station, TX 77843-3368, USA} \email{liuwencai1226@gmail.com}

\keywords{Multi-scale analysis, large deviation theorem, discrepancy, semi-algebraic sets, Cartan's techniques, Anderson localization.}
\subjclass[2010]{  81Q10 (primary); 82B44, 37C55 (secondary)}

\theoremstyle{plain}
\newtheorem{theorem}{Theorem}[section]
\newtheorem{corollary}[theorem]{Corollary}
\newtheorem{lemma}[theorem]{Lemma}

\newcommand{\C}{\mathbb{C}}
\newcommand{\R}{\mathbb{R}}
\newcommand{\T}{\mathbb{T}}
\newcommand{\Z}{\mathbb{Z}}
\newcommand{\N}{\mathbb{N}}
\theoremstyle{definition}

\newtheorem{remark}{Remark}

\begin{document}
	

\maketitle
\renewcommand{\abstractname}{Abstract}
\begin{abstract}
	We provide quantitative inductive estimates for Green's functions  of  matrices  with (sub)expoentially decaying off diagonal entries  in higher dimensions.
	Together with  Cartan's estimates and discrepancy estimates,
	we establish  explicit   bounds for the large deviation theorem  for  non-self-adjoint     Toeplitz operators.
	As applications,  we obtain  the modulus of    continuity  of  the integrated density of states  with explicit  bounds and the  pure point spectrum property for analytic quasi-periodic operators.
	Moreover, our inductions   are self-improved and work for perturbations with low complexity interactions.
\end{abstract}

\section{Introduction}

The  dynamics and spectral theory for  quasi-periodic  operators  have  been made  significant progress
in the last 40 years, through earlier  perturbative methods \cite{ds75,eli,fsw,si87,cs89,MP84}, and then
non-perturbative methods by controlling   
  Green's functions/transfer matrices  \cite{j94,j,bj,bbook,bg20,bgscmp} or  by   reducibility \cite{hy12,afk11}.   The case of
one dimensional  lattice   and one frequency potential, has been well understood  for both small and large coupling constants, with the recent discovery  of
global theory \cite{global} and universal  structure  \cite{jl1,jl2}.  In particular,  remarkable developments have been achieved for  several   models  motivated by physics: the almost Mathieu operator (the Harper's model), the extended Harper's  model and the
Maryland model \cite{ayz,jl1,jl2,jlcpam,sim85,hj17,jm12,j,lyjst,lyjfa,ten,liuetds,jkl,jk19,jpoint,ad08,ab08,ak06,ls19,ls19,jkjems,jzjems}.
We refer readers to
\cite{jmarx,you} and
references therein for more details.

Problems are known to be much more complicated if one increases the underlying dimension $b$ of the torus or the dimension $d$ of the lattice.
The higher dimension picture is still far from clear. For the one dimensional lattice $d=1$ and multi-frequencies $b\geq 1$,    some special cases have been studied   by transfer matrices or Schr\"odinger cocycles \cite{gsv,gs01,bbook,dgsv,amor,CCYZ19,eli}.
 The first
multi-dimensional localization  result was obtained by perturbative (KAM)
methods by Chulaevsky-Dinaburg for operators on lattices  $\Z^d$ and torus  $\T$
for   arbitrary $d$ \cite{cd93}.
 Bourgain-Goldstein-Schlag developed a celebrated  method  in the spirit of non-perturbative approaches from Bourgain-Goldstein \cite{bg20}  
  to handle the  two-dimension
and two-frequency case \cite{bgs} ($b=d=2$) and  established the Anderson localization for  large coupling constants.  This is the first higher dimension lattice and multi-frequency result.
 Moreover, the large deviation theorem in \cite{bgs}, which is a key ingredient to prove the Anderson localization,  is purely arithmetic in the sense that removed  sets of  frequencies are independent of  the potential. 
Roughly speaking,  by imposing some purely arithmetic condition on $(\omega_1,\omega_2)\in \R^2$,  for  any 
algebraic curve $\Gamma\subset [0,1]^2$ with  degree  at most  $N^C$,
the number of lattice points 
\begin{equation}\label{gsub}
\{(n_1,n_2)\in \Z^2: |n_1|\leq N, |n_2|\leq N , (n_1\omega_1,n_2\omega_2)\mod \Z^2\in  \Gamma_{\tau} \}   
\end{equation}
is bounded by $N^{1-\delta}$ for some $\delta>0$, where $\Gamma_\tau$ is the $e^{-N^\tau}$ neighborhood of $\Gamma$.
The quantity $N^{1-\delta}$ is referred  to as the sublinear bound.
It is still open whether the analogy for $d\geq 3$  is true or not.  

In \cite{bgs},  Bourgain   developed a new
scheme  to prove the large deviation theorem 
 for  arbitrary $b=d$ \cite{gafa} by a delicate study of the semi-algebraic sets.
 Jitomirskaya-Liu-Shi extended
Bourgain's result to the case of arbitrary $b$ and $d$ \cite{jls}.   However,  the removed set  of frequencies in \cite{gafa,jls} depends on the potential.


  Bourgain, Goldstein and Schlag \cite{bgs} mentioned  that  the sub-linear bound \eqref{gsub}  is the only obstruction to  establish an arithmetic version of the  large deviation theorem   in higher dimensions. However, there is no detailed proof available yet. 
 Our first goal of this paper is  to provide  such a    proof.  Moreover,  we  are going to establish the quantitative version of the  main results in  \cite{bgs} with  generalizations, in particular it   can be applied to  quasi-periodic operators on arbitrary lattices $\Z^d$ driven by any dynamics on tori $\Z^b$ under the assumption on sub-linear bounds.

Instead of      Laplacians or  long range operators,
we will  study  Toeplitz  matrices  with (sub)exponentially decaying off diagonal entries.  
Among all the motivations of our generalizations, we want to highlight one.
Anderson localization receives  a lot of  attentions  from both mathematics and physics. The approach to establish Anderson localization for quasi-periodic operators with analytic potentials turns out to be a breakthrough component  to construct quasi-periodic solutions for nonlinear Schr\"odinger equations and nonlinear wave equations \cite{bbook,wangduke}.
It is  known that the quasi-periodic solutions in PDEs  are only sub-exponentially, not exponentially decaying \cite{b98,bbook,wangduke}. 
Therefore, the  (sub)exponentially decaying  matrices are   more   natural settings in PDEs.  

In our arguments, the matrices are    not necessarily  self-adjoint and every entry   of  the matrices is  allowed to be  a function.
	For $d\geq 2$, this is the   first time to  study operators that beyonds  long range  cases. 
For $d=1$,   our assumptions are  weaker than Bourgain's \cite{bkick}. See Remark \ref{rebound} for details. Moreover, our arguments hold under perturbations  with low complexity.

Our proof is definitely inspired by \cite{bgs}.  However, there are a lot of  important ingredients  being added into the arguments to make it quantitative  in our more general settings. 
 Moreover,
we  significantly simplify the arguments even for the case appearing in \cite{bgs}.
The analysis of \cite{bgs} required dealing with many different types of elementary regions, say rectangles and $L$-shapes in $\Z^2$.  We largely  reduced the  elementary regions  to be  square related. See Fig.1. 
Two novelties  are added  here. Firstly, we introduce the concept of width of   subsets of lattices. In our augments, we always keep the involved regions $\Lambda$ having large width  so that every lattice point in $\Lambda$ can be covered by a square related elementary region with presetting size contained in  $\Lambda$. For example, the region like   Fig.2 was not allowed because the width determined by the distance between $B$ and $C$ is too small.  Secondly, we reconstruct the exhaustion of $x$ in every elementary region. In our new construction, the  annuli with small width are absorbed into  bigger ones. See Fig.3.

There are several other  technical improvements in this paper, which we believe to be of independent interest.
For example, we estabish the Cartan's estimates for non-self-adjoint matrices.

We will prove a  quantitatively  inductive theorem  about the   Green's functions     in higher dimensions as stated in Theorem \ref{thmmul}.  
This is a deterministic statement, 
which  can be  applied to study operators  even without dynamics.  
Based on  matrix-valued Cartan-type theorem (estimates on subharmonic functions)  in \cite{bgs} with further developments in  \cite{bbook,gs08,jls},
we will establish the measure estimates in Theorem \ref{thmmu2}. 
Imposing proper dynamics on  tori,
the  quantitative inductive  estimate for Green's functions is obtained (Theorem \ref{thmu}). 
Moreover,     the  relation among all  constants and parameters is displayed clearly so that   the whole picture  becomes extremely transparent.
We will   see  how  arithmetic conditions on  frequencies effect  the discrepancy, how structures of  semi-algebraic sets effect   the number of bad Green's functions, and how  the  dimensions of lattices and frequencies  contribute to  bounds.

Finally, we want to talk about the applications.
As far as we know, there is no    explicit bound yet for the  large deviation theorem except for the case $d=1$  and $b=1,2$.
Our approaches (Theorems \ref{thmmul}, \ref{thmmu2} and \ref{thmu}) are the first time to establish the explicit bounds in higher dimensions  and multi-frequencies.
We show  that in the arithmetic sense, for  $d=1$ and any $b$, the bound is   arbitrarily close to $\frac{1}{b^3}$ for shift dynamics and $\frac{1}{4^{b-1}b^3}$ for skew-shift dynamics.
 For $b=1$ and arbitrary $d$, we show that  the bound is   arbitrarily close to $1$. 
 
 Another application  we  want to mention     is the regularity of the integrated density of states (IDS) of quasi-periodic operators.
 The log-H\"older continuity of the integrated density of states  is quite general \cite{cbcmp,bkinv}.  
 The H\"older continuity  in one dimensional settings was well established \cite{blmp2000,gs08,aj10,amor,bbook,gs01,ly15,ajcmp,CCYZ19,XZhao20,hzimrn} for both large and small coupling constants.  
 What we will investigate in this paper  is the modulus of continuity $f(x)=e^{-\kappa|\log x|^\tau}$. 
 Unfortunately, like   the large deviation theorem,  except for  the case $d=1$  and $b=1,2$, there are no explicit bounds  of $\tau$  in the region of large coupling constants.
 Based on the ingredients from \cite{blmp2000,schcmp} and the  large deviation theorem,  the   modulus of continuity  of the integrated density of states  with   explicit estimates  will be obtained in Theorem \ref{thmholder}.

\section{Main results}

Let $A$ be a (operator) matrix on $\ell^2(\Z^d)$ satisfying,
\begin{equation}\label{GO}
|A(n,n')|\leq   Ke^{-c_1|n-n'|^{\tilde{\sigma}}},  K>0, c_1>0,0<\tilde{\sigma}\leq 1,
\end{equation}
 where $|n|:=\max\limits_{1\leq i\leq d}|n_i|$ for $n=(n_1,n_2,\cdots,n_d)\in \Z^d$.
 We  say that  the off diagonal entries of $A$ are  subexponentially decaying  if $A$ satisfies \eqref{GO}. Sometimes, we just say 
 $A$ is  subexponentially decaying for simplicity.

For $d=1$,   the elementary region of size $N$ centered at 0 is given by
\begin{equation*}
  Q_N=[-N,N].
\end{equation*}

For $d\geq 2$, denote by $Q_N$ an elementary region of size $N$ centered at 0, which is one of the following regions,
\begin{equation*}
  Q_N=[-N,N]^d
\end{equation*}
or
$$Q_N=[-N,N]^d\setminus\{n\in\mathbb{Z}^d: \ n_i\varsigma_i 0, 1\leq i\leq d\},$$
where  for $ i=1,2,\cdots,d$, $ \varsigma_i\in \{<,>,\emptyset\}^{d}$ and at least two $ \varsigma_i$  are not $\emptyset$.

Denote by $\mathcal{E}_N^{0}$ the set of all elementary regions of size $N$ centered at 0. Let $\mathcal{E}_N$ be the set of all translates of  elementary regions  with center at 0, namely,
$$\mathcal{E}_N:=\{n+Q_N:n\in\mathbb{Z}^d,Q_N\in \mathcal{E}_N^{0}\}.$$
We call  elements in $\mathcal{E}_N$ elementary regions.

{\bf Example 1: } For $d=2$, there are five types of elementary regions.
\begin{center}
\begin{tikzpicture}[]
\draw[](0,0)rectangle (2,2);
\draw (3,0)--(5,0)--(5,1)--(4,1)--(4,2)--(3,2)--(3,0);
\draw (6,0)--(8,0)--(8,2)--(7,2)--(7,1)--(6,1)--(6,0);
\draw (10,1)--(10,0)--(11,0)--(11,2)--(9,2)--(9,1)--(10,1);
\draw (12,0)--(13,0)--(13,1)--(14,1)--(14,2)--(12,2)--(12,0);
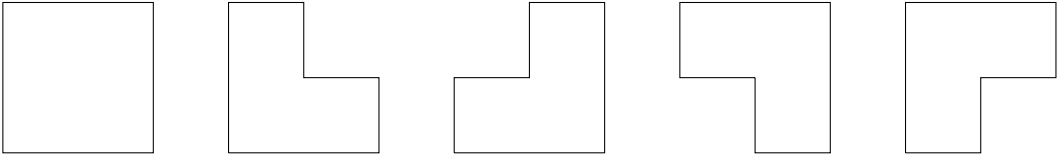
\draw (7,0) node[below]   {Fig.1: elementary regions in $\Z^2$};



\end{tikzpicture}
\end{center}

The  width  of a   subset  $\Lambda\subset \Z^d$, is defined by maximum
 $M\in \N$  such that  for any  $n\in \Lambda$, there exists  $\hat{M}\in \mathcal{E}_M$ such that
\begin{equation*}
  n\in \hat{M} \subset \Lambda
\end{equation*}
and
\begin{equation*}
  \text{ dist }(n,\Lambda\backslash \hat{M})\geq M/2.
\end{equation*}
{\bf Example 2:}  In Fig.2, the width of $\Lambda$ is determined by the distance between B and C.
\begin{center}
\begin{tikzpicture}[]
\draw (0,0)--(4,0)--(4,-2)--(4.3,-2)--(4.3,-4)--(0,-4)--(0,0);
\draw (2,-2) node[below]   {$\Lambda$};
\draw (3.85,-2) node[]   {B};
\draw (4.45,-2) node[]   {C};
\draw (2,-4) node[below]   {Fig.2: a region with small  width};




\end{tikzpicture}
\end{center}


A generalized  elementary region is defined to be a subset $\Lambda\subset \Z^d$ of the form
\begin{equation*}
  \Lambda:= R\backslash(R+z),
\end{equation*}
where $z\in\Z^d$ is arbitrary and $R$ is a rectangle,
\begin{equation*}
  R=\{n=(n_1,n_2,\cdots,n_d)\in \Z^d: |n_1-n_1^\prime|\leq M_1, \cdots,|n_d-n_d^\prime|\leq M_d\}.
\end{equation*}
For $ \Lambda\subset\mathbb{Z}^d$,  we introduce its diameter,
$$\mathrm{diam}(\Lambda)=\sup_{n,n'\in \Lambda}|n-n'|.$$

 Denote by $\mathcal{R}_N$
all  generalized elementary regions with diameter less than or equal  to $N$.
Denote by $\mathcal{R}_N^M$
all   generalized elementary regions in $\mathcal{R}_N$ with width larger than or equal to $M$.
For $\Lambda\subset\mathbb{Z}^d$,  let $R_{\Lambda}$ be the restriction operator, i.e.,  $(R_{\Lambda}u)(n)=u(n)$ for $n\in \Lambda$, and
$(R_{\Lambda}u)(n)=0$ for $n\notin \Lambda$.

We say an elementary region $ \Lambda\in \mathcal{E}_{N^\prime}$ is  in class {\it G} (Good)  if
\begin{equation}\label{ggood}
  |(R_{\Lambda}A  R_{\Lambda})^{-1}(n,n^\prime)|\leq  e^{-c_2|n-n^\prime|^{\tilde{\sigma}}}, \text{ for } |n-n^\prime|\geq \frac{N^\prime}{10},
\end{equation}
where $0<c_2\leq \frac{5^{\tilde{\sigma}}-1}{5^{\tilde{\sigma}}}c_1$ and $0<\tilde{\sigma}\leq1$.
We mentioned that the upper bound $\frac{5^{\tilde{\sigma}}-1}{5^{\tilde{\sigma}}}c_1$ is chosen for technical convenience. See \eqref{gdec94} for the explanation.

Denote by   $\lfloor x\rfloor$  the largest integer smaller than  or equal to $x$.
\begin{theorem}\label{thmmul}
Assume $A$ satisfies \eqref{GO}.
Let $\varsigma,\sigma,\xi\in(0,1)$ and $\sigma<\tilde{\sigma}\leq 1$.
Let $\tilde{\Lambda}_0\in \mathcal{E}_N$ be an elementary region with the property that for all $ \Lambda\subset \tilde{\Lambda}_0$, $\Lambda\in \mathcal{R}_L^{N^\xi}$ with $N^{\xi}\leq L\leq N$, the Green's function $(R_{\Lambda}A  R_{\Lambda})^{-1}$ satisfies 
\begin{equation}\label{gboundnov26}
  ||(R_{\Lambda}A  R_{\Lambda})^{-1}||\leq e^{L^{\sigma}}.
\end{equation}
Assume that for any family $\mathcal{F}$ of pairwise disjoint elementary regions in $\tilde{\Lambda}_0$ with size $M=\lfloor N^\xi \rfloor$,
\begin{equation}\label{gnbad}
  \# \{ \Lambda \in \mathcal{F}: \Lambda \text { is not in class G } \}\leq \frac{N^{\varsigma}}{N^\xi}.
\end{equation}
Then
for large $N$ (depending on $K,\varsigma,\sigma,\tilde{\sigma},\xi,c_1$ and the lower bound of $c_2$),
\begin{equation}\label{gthm1}
  |(R_{\tilde{\Lambda}_0}A  R_{\tilde{\Lambda}_0})^{-1}(n,n^\prime)|\leq  e^{-(c_2-N^{-\vartheta})|n-n^\prime|^{\tilde{\sigma}}},\text{ for } |n-n^\prime|\geq \frac{N}{10},
\end{equation}
where $\vartheta=\vartheta(\sigma,\tilde{\sigma},\xi,\varsigma)>0$.
\end{theorem}
Here are several comments   about Theorem \ref{thmmul}.
\begin{remark}\label{remul}
\begin{enumerate}
 \item For $d=1$ and $\tilde{\sigma}=1$, a  similar statement    was proved by Bourgain \cite{bkick}.
  For $d=2$ and $\tilde{\sigma}=1$,  a similar statement  was proved for the particular case where $A$ is given by the discrete Laplacian \cite{bgs}. 
  \item   The   statement in Theorem   \ref{thmmul}  is a robust approach to deal with the spectral theory for quasi-periodic operators and also  the construction of  quasi-periodic solutions for nonlinear Schr\"odinger/wave equations. See   \cite{bwjems,bwcmp,bgs,bkick,bgscmp} for applications.  Some particular cases of Theorem \ref{thmmul}  have been used  as  ingredients  to construct quasi-periodic solutions for  PDEs  and   have been stated in  \cite{bwjems,bwcmp,wangduke} without detailed proof.  There are no explicit bound estimates in their arguments either.
       

  \item 
  In applications, $\varsigma$ is chosen to be arbitrarily close to $1$, namely $\varsigma=1-\varepsilon$ with arbitrarily  small $\varepsilon>0$. Then the upper bound  in \eqref{gnbad}   equals   $N^{1-\xi-\varepsilon}$.
 Theorem \ref{thmmul} says that the ``goodness" of Green's functions at small size $N^{\xi}$ will ensure the ``goodness" of Green's functions at larger size $N$
  under the following two conditions:
  \begin{itemize}
       \item  The number of bad Green's functions of size  $N^{\xi}$ in   $[-N,N]^d$ is less than $N^{1-\xi-\varepsilon}$ (referred to  as the sub-linear bound).
       \item  The   Green's functions can not be ``super bad" in the sense that they  are controlled by \eqref{gboundnov26}. The upper bound $e^{L^{\sigma}}$ with $\sigma<1$ is referred to as the sub-exponential bound.
     \end{itemize}


\end{enumerate}
\end{remark}
Let $b=\sum_{i=1}^k b_i$, where $b_i\in \N$. Let $x=(x_1,x_2,\cdots,x_k)$, where $x_i\in \T^{b_i}=(\R/\Z)^{b_i}$, $i=1,2,\cdots,k$.
For any $x\in\mathbb{T}^b$ and $1\leq i\leq k$, let $${x}_i^\neg=(x_1,\cdots,x_{i-1},x_{i+1}\cdots,x_k)\in\mathbb{T}^{b-b_j}.$$
For any $y\in\mathbb{T}^{d_1}$ and $X\subset\mathbb{T}^{d_1+d_2}$,  denote   the $y$-section of $X$:
$$X(y):=\{z\in\mathbb{T}^{d_2}:\  (y,z)\in X\}.$$
Write  ${\rm Leb}(S)$ for the Lebesgue measure.

Assume each element of  the operator $A$  is a function on $ \T^b$. Sometimes, we  indicate the dependence and denote by the element   $A(x;n,n^\prime)$.
Assume every element $A(z;n,n^\prime)$   is analytic in the strip $\{z\in \C^b:|\Im z|\leq \rho\}$, $\rho>0$,  and  satisfies for any $n,n^\prime\in\Z^d$ and  $x\in\T^b$,
\begin{equation}\label{GOnew}
|A(x;n,n')|\leq  K e^{-c_1|n-n'|^{\tilde{\sigma}}}, K>0,  c_1>0,0<\tilde{\sigma}\leq1.
\end{equation}

Assume that  there exists $K_1>1$ such that for any $x\in\T^b$ and $z\in \{z\in\C^b:|\Im z|\leq \rho\}$ with $||x-z||\leq e^{ -\left(\log  (|n|+|n^\prime|+2)\right)^{K_1}}$, 
\begin{equation}\label{glc1}
|A(x;n,n')-A(z;n,n')|\leq K ||x-z||^\gamma,
\end{equation}
where $||z||={\rm dist}(z,\Z^b)$.

{\bf Example 3.} 
If $A$ satisfies \eqref{GOnew} and for any $n,n^\prime\in\Z^d$, $A(x;n,n^\prime)$  is a trigonometric  polynomial of degree  at most  $e^{ \left(\log  (|n|+|n^\prime|+2)\right)^{K_1}}$, then \eqref{glc1} holds.

 We say  an elementary region $ \Lambda\in \mathcal{E}_{N}$ is  in class {\it SG$_{N}$} (strongly good with size $N$) if
 \begin{equation}\label{ggoodt1}
 ||R_{\Lambda}A  R_{\Lambda})^{-1}||\leq  e^{N^{\sigma}},
 \end{equation}
 and
\begin{equation}\label{ggoodt2}
  |(R_{\Lambda}A  R_{\Lambda})^{-1}(n,n^\prime)|\leq  e^{-c_2|n-n^\prime|^{\tilde{\sigma}}}, \text{ for } |n-n^\prime|\geq \frac{N}{10},
\end{equation}
where $0<c_2\leq\frac{5^{\tilde{\sigma}}-1}{5^{\tilde{\sigma}}}c_1$ and $0<\sigma<\tilde{\sigma}\leq1$. When there is no confusion, we drop  the dependence of $N$ from  the notation {\it SG$_{N}$}.
\begin{theorem}\label{thmmu2}
Assume  $A$ satisfies \eqref{GOnew} and \eqref{glc1}.
Fix $\sigma,\delta,\tilde{\sigma},\zeta\in (0,1)$ and $ \mu\in(1-\delta,1)$, $\sigma<\tilde{\sigma}$.
Suppose $\mathcal{R}\subset[-N_3,N_3]^d$ has width at least $N_2$.
For $x\in \T^b$, define $\mathcal{B}_{\mathcal{R}}(x)$ as
\begin{equation*}
  \mathcal{B}_{\mathcal{R}}(x)=\{ n\in \mathcal{R}: \text{ there exists } Q_{N_1}\in  \mathcal{E}_{N_1}^{0} \text{ such that } n+Q_{N_1}\notin \text{{\it SG}}_{N_1}\}
\end{equation*}
Assume that for any $x\in\T^b$,
\begin{equation}\label{gsublinear}
  \#  \mathcal{B}_{\mathcal{R}}(x)\leq L^{1-\delta}.
\end{equation}
Assume that  there exists a subset $X_{N_2}\subset \T^b$, such that
\begin{equation}\label{Gstarnov}
 \sup_{1\leq i\leq k,x_i^\neg\in \T^{b-b_i}}\mathrm{Leb}( X_{N_2}(x_i^\neg))\leq e^{-{N_2}^{\zeta}},\end{equation}
and  for  any    $ Q_{N_2}\in \mathcal{E}_{N_2}^0$, $x\notin  X_{N_2}$ and $n\in \mathcal{R}$, the region $n+Q_{N_2}$  is  in class {\it SG$_{N_2}$}.
Let
\begin{equation*}
  \tilde{X}_{\mathcal{R}}(x)=\{x\in \T^b: ||(R_{\mathcal{R}} A(x)R_{\mathcal{R}})^{-1}||\geq e^{L^{\mu}}  \}.
\end{equation*}
Suppose $ N_3\leq e^{ N_1^{\frac{1}{2K_1}}}$, $N_2\geq N_1^{\frac{2}{\zeta}}$ and $L\geq N_2^{\frac{2d+b+2}{\mu-1+\delta}}$.
 Then  there exists $N_0=N_0(K_1,K, c_1, c_2,\tilde{\sigma}, \sigma,\delta,\gamma,\rho,\mu)$\footnote{It depends on the lower bound of $c_2$.}
 such that for any $N_1\geq N_0$ and $i=1,2,\cdots,k$,
\begin{equation}\label{Gstar}
 \sup_{x_i^\neg\in \T^{b-b_i}}\mathrm{Leb}( \tilde{X}_{\mathcal{R}}(x_i^\neg))\leq  e^{-\left(\frac{{L}^{\mu-1+\delta}}{N_2^{2d+b+2}}\right)^{1/b_i}}.
 \end{equation}
\end{theorem}

Let $f$ be a function from $\Z^d\times\T^b$ to $\T^b$.
Assume for any $m_1,m_2,\cdots,m_d\in \Z^d$ and  $n_1,n_2,\cdots,n_d\in \Z^d$,
\begin{equation*}
 f({m_1+n_1},{m_2+n_2},\cdots, {m_d+n_d},x)= f({m_1},{m_2},\cdots, {m_d}, f({n_1},{n_2},\cdots,{n_d},x)).
\end{equation*}
Sometimes, we write down $f^n(x)$ for $f(n,x)$ for convenience, where $n\in \Z^d$ and $x\in \T^b$.
We say $A$ is a Toeplitz  (operator) matrix on $\ell^2(\Z^d)$ with respect to $f$,  if
\begin{equation}\label{gdec72}
  A(x;n+k,n^\prime+k)= A(f^k(x);n,n^\prime), 
\end{equation}
for any $n\in \Z^d, n^\prime\in \Z^d$ and $k\in \Z^d.$ We note that $A$ is not necessarily self-adjoint.



 We say the Green's function  of an operator $A(x)$  satisfies property $P$ with parameters  $(\mu,\zeta,c_2)$ at size $N$ if the following statement is true:
     there exists a  subset $X_N\subset \mathbb{T}^b$  such that
 \begin{equation*}
 \sup_{1\leq i\leq k,x_i^\neg\in \T^{b-b_i}}\mathrm{Leb}(X_N(x_i^\neg))\leq e^{-{N}^{\zeta}},
 \end{equation*}
and for any $x\notin  X_N \mod \Z^b$ and $Q_N\in \mathcal{E}_N^0$,
\begin{eqnarray*}
 || (R_{Q_N}A(x)R_{Q_N})^{-1} ||&\leq& e^{N^{\mu}} ,\\
  |(R_{Q_N}AR_{Q_N})^{-1}(x;n,n^\prime)| &\leq& e^{-c_2|n-n^\prime|^{\tilde{\sigma}}}, \text{ for } |n-n^\prime|\geq \frac{N}{10}.
\end{eqnarray*}

\begin{theorem}\label{thmu}
	Assume $A(x)$ satisfies  \eqref{GOnew}, \eqref{glc1} and \eqref{gdec72}, and 
	$$ 0<	c_2<(1-5^{-\tilde{\sigma}})c_1,1-\delta<\sigma<\tilde{\sigma}\leq 1, \delta>\iota>0,\text{ and } 0<\mu<\tilde{\sigma}.$$ 
	Let $c=\frac{1}{2}\min\{\frac{1}{  K_1},\tilde{\sigma}\}$.
	Fix any sufficiently small $\varepsilon>0$.
	There exists a large constant $C$ depending on all parameters   such that the following statements are true.
	Let $N_1$ be sufficiently large, $N_2\in [N_1^C, e^{N_1^{c/2}}]$ and $N_3\in [N_2^C,e^{N_1^{c}}]$.
	Assume that  the  Green's function  satisfies the property ${P}$ with parameters  $(\mu,\zeta,c_2)$ at sizes $N_1$ and $N_2$.
	Assume for any $L\in [N_3^{\delta-\iota},N_3]$ and any $x\in\T^b$,
	\begin{equation}\label{gdec5}
	\#\{ n\in \Z^d: |n|\leq L, f(n,x)\in X_{N_1}\mod \Z^b\}\leq L^{1-\delta}.
	\end{equation}
	Then
	there exists  ${X}_{N_3}\subset \mathbb{T}^b$  such that
	\begin{equation}\label{gdec3u}
	\sup_{1\leq i\leq k,x_i^\neg\in \T^{b-b_i}}\mathrm{Leb}(X_{N_3}(x_i^\neg))\leq e^{-{N_3}^{\frac{\sigma-1}{b_i}\delta+\frac{\delta^2}{b_i}-\varepsilon}},
	\end{equation}
	and for any $x\notin  X_{N_3}$ and $Q_{N_3}\in \mathcal{E}_{N_3}^0$,
	\begin{eqnarray}
	|| (R_{Q_{N_3}}A(x)R_{Q_{N_3}})^{-1} ||&\leq& e^{N_3^{\sigma}},\label{equfeb6}
	\end{eqnarray}
	and for $|n-n^\prime|\geq \frac{N_3}{10}$,
	\begin{eqnarray}
	|(R_{Q_{N_3}}AR_{Q_{N_3}})^{-1}(x;n,n^\prime)| &\leq& e^{-(c_2-2N_1^{-\vartheta_1}-N_{3}^{-\vartheta_2})
		|n-n^\prime|^{^{\tilde{\sigma}}}},\label{gdec49}
	\end{eqnarray}
where $\vartheta_1=\vartheta_1(\tilde{\sigma},\mu,c)$ and $\vartheta_2=\vartheta_2(\tilde{\sigma},\sigma,\delta,\varepsilon)$.
\end{theorem}

Our Theorems  work  for  Topelitz matrices with  low complexity interactions.
Let $U$ be  an operator on  $\ell^2(\Z^d)$ satisfying
\begin{equation*} 
|U(n,n')|\leq  K e^{-c_1|n-n'|^{\tilde{\sigma}}}.
\end{equation*}  
Given $m\in \Z^d$,    define the  operator $U^{m}$ by 
$$U^{m}(n,n^\prime)=U(m+n,m+n^\prime),n\in\Z^d,n^\prime\in\Z^d.$$
We say $U$ has low complexity if  there exists $0<a<1$ such that  for any $N>1$,
\begin{equation}\label{comu}
\#\{ R_{Q_N}U^m R_{Q_N}: m\in\Z^d,Q_{N}\in \mathcal{E}_{N}^0 \}\leq K e^{N^{a}}.
\end{equation}

For any $m\in \Z^d$, denote by
\begin{equation}\label{Aint}
\tilde{A}^m(x;n,n^\prime)=A(x;n,n^\prime)+U^m(n,n^\prime).
\end{equation}
 We say that  the Green's function  of an operator  ${A}(x)$  satisfies property $\tilde{P}$ with parameters  $(\mu,\zeta,c_2)$ at size $N$ if the following statement is true:
there exists a  set $X_N\subset \mathbb{T}^b$  such that
\begin{equation*}
\sup_{1\leq i\leq k,x_i^\neg\in \T^{b-b_i}}\mathrm{Leb}(X_N(x_i^\neg))\leq e^{-{N}^{\zeta}},
\end{equation*}
and for any $x\notin  X_N \mod \Z^b$, $m\in\Z^d$, and $Q_N\in \mathcal{E}_N^0$
\begin{eqnarray*}
	|| (R_{Q_N}\tilde{A}^m(x)R_{Q_N})^{-1} ||&\leq& e^{N^{\mu}} ,\\
	|(R_{Q_N}\tilde{A}^mR_{Q_N})^{-1}(x;n,n^\prime)| &\leq& e^{-c_2|n-n^\prime|^{\tilde{\sigma}}}, \text{ for } |n-n^\prime|\geq \frac{N}{10}.
\end{eqnarray*}
We have
\begin{theorem}\label{thmunew}
	Assume $A(x)$ satisfies  \eqref{GOnew}, \eqref{glc1} and \eqref{gdec72}, $U$ has low complexity,  
	$$ 0<	c_2<(1-5^{-\tilde{\sigma}})c_1,1-\delta<\sigma<\tilde{\sigma}\leq 1, \delta>\iota>0,  0<\mu<\tilde{\sigma}, $$
	and 
	\begin{equation}\label{gcom}
	a\leq \frac{1}{2}\min_i\left\{\frac{\sigma-1}{b_i}\delta+\frac{\delta^2}{b_i}\right\}.
	\end{equation}	
	 Let $\tilde{A}^m$ be given by \eqref{Aint} and $c=\frac{1}{2}\min\{\frac{1}{  K_1},\tilde{\sigma}\}$.
Fix any sufficiently small $\varepsilon>0$.
Then
there exists a large constant $C$ depending on all parameters   such that the following statements are true.
Let $N_1$ be sufficiently large,  $N_2\in [N_1^C, e^{N_1^{c/2}}]$ and   $N_3\in [N_2^C,e^{N_1^{c}}]$.
Assume the  Green's function  satisfies the property $\tilde{P}$ with parameters  $(\mu,\zeta,c_2)$ at sizes $N_1$ and $N_2$.
Assume for any $L\in [N_3^{\delta-\iota},N_3]$ and any $x\in\T^b$,
	\begin{equation*}
	\#\{ n\in \Z^d: |n|\leq L, f(n,x)\in X_{N_1}\mod \Z^b\}\leq L^{1-\delta}.
	\end{equation*}
	Then
	there exists a  subset ${X}_{N_3}\subset \mathbb{T}^b$  such that
	\begin{equation*}
	\sup_{1\leq i \leq k,x_i^\neg\in \T^{b-b_i}}\mathrm{Leb}(X_{N_3}(x_i^\neg))\leq e^{-{N_3}^{\frac{\sigma-1}{b_i}\delta+\frac{\delta^2}{b_i}-\varepsilon}},
	\end{equation*}
	and for any $x\notin  X_{N_3}$, $m\in \Z^d$ and  $Q_{N_3}\in \mathcal{E}_{N_3}^0$,
	\begin{eqnarray*}
	|| (R_{Q_{N_3}}\tilde{A}^m(x)R_{Q_{N_3}})^{-1} ||&\leq& e^{N_3^{\sigma}},
	\end{eqnarray*}
and  for $|n-n^\prime|\geq \frac{N_3}{10}$, 
	\begin{eqnarray*}
	|(R_{Q_{N_3}}\tilde{A}^mR_{Q_{N_3}})^{-1}(x;n,n^\prime)| &\leq& e^{-(c_2-N_1^{-\vartheta_1}-N_{3}^{-\vartheta_2})|n-n^\prime|^{^{\tilde{\sigma}}}},
\end{eqnarray*}
where $\vartheta_1=\vartheta_1(\tilde{\sigma},\mu,c)$ and $\vartheta_2=\vartheta_2(\tilde{\sigma},\sigma,\delta,\varepsilon)$.
\end{theorem}

\begin{remark}\label{redec81}
\begin{enumerate}

\item    Theorem \ref{thmu} improves the parameters from $(\mu,\zeta,c_2)$ to $$(\sigma,\frac{\sigma-1}{b_i}\delta+\frac{\delta^2}{b_i}-\varepsilon,c_2-N_1^{-\vartheta_1}-N_{3}^{-\vartheta_2}).$$
Theorem \ref{thmu} gives us opportunities to combine perturbative approaches with non-perturbative approaches. After establishing the property  P  for initial scales by non-perturbative methods, we can adapt the parameters to establish property  P with explicit  bounds for larger scales. See Theorems \ref{thmapp1},  \ref{thmapp1'} and \ref{coroapp1},  and Corollaries  \ref{coroapp10}, \ref{thmapp1'new} and \ref{coroapp1new} for examples.

  \item  Roughly speaking
   Theorem \ref{thmu} says that  under the assumption on the sublinear bound,  the large deviation theorem  at sizes   $N=N_1$ and $N=N_2$ will  ensure the large deviation theorem  at size $N=N_3$.

\end{enumerate}
\end{remark}

We are going to discuss the modulus of continuity of the integrated density of states (IDS).
  In order to make it as general as possible, we do not require the existence of the integrated density of states first. 
Let $E_1< E_2$ and define
\begin{equation}\label{ggids}
  k(x,E_1,E_2)=\limsup_{N\to\infty}\frac{1}{(2N+1)^d}\#\{\text{ eigenvalues of } R_{{[-N,N]^d}}A(x)R_{{[-N,N]^d}} \text{ in  } [E_1,E_2]\}.
\end{equation}
Fix $x\in \T^b$.
Assume  for any measurable set $\mathcal{S}\subset\T^b$, we have
\begin{equation}\label{gx}
  \limsup_{N\to\infty}\frac{1}{(2N+1)^d}\#\{n\in\Z^d: |n|\leq N,  f(n_1,n_2,\cdots,n_d,x)\in \mathcal{S}\}\leq {\rm Leb}(\mathcal{S}).
\end{equation}
For an operator $A(x)$   on $\ell^2(\Z^d)$,
denote by the energy dependent  Green's functions
\begin{equation}\label{ge}
  G_{\Lambda}(E,x)=(R_{\Lambda} (A(x)-E)R_{\Lambda})^{-1}.
\end{equation}
 Instead of $ G_{\Lambda}(E,x) $,  we will  write $G_{\Lambda}$, $G_{\Lambda}(E)$, or $G_{\Lambda}(x)$ when there is no ambiguity.
We will write
 $G_{\Lambda}(n,n^\prime)$, $G_{\Lambda}(E;n,n^\prime)$, $G_{\Lambda}(x;n,n^\prime)$, or $G_{\Lambda}(E,x;n,n^\prime)$  for the element of  matrices.


\begin{theorem}\label{thmholder}
Assume $A(x)$ is a Toeplitz  (operator) matrix on $\ell^2(\Z^d)$ with respect to $f$ in the sense of \eqref{gdec72}.
Let $  \zeta\in(0,1)$ and $0<\sigma<\tilde{\sigma}\leq 1$.
Assume  for any $E\in \R$, there exists a set $X_N\subset \T^b $ such that
\begin{equation*}
  {\rm Leb}(X_N)\leq e^{-N^{\zeta}}
\end{equation*}
and for any $x\notin X_N$ and  any $Q_N\in \mathcal{E}_{N}^0$,
\begin{eqnarray*}
  ||G_{Q_N}(E,x)|| &\leq & e^{N^{\sigma}} \\
  |G_{Q_N}(E,x;n,n^\prime)| &\leq& e^{-c|n-n^\prime|^{^{\tilde{\sigma}}}} \text{ for } |n-n^\prime|\geq \frac{N}{10},
\end{eqnarray*}
where $c>0$.
Assume \eqref{gx} holds for  some $x_0\in \T^b$.
Then  for any $\varepsilon>0$, we have
\begin{equation*}
 | k(x_0,E_1,E_2)|\leq e^{-|\log|E_1-E_2||^{\frac{\zeta}{\sigma}-\varepsilon}},
\end{equation*}
provided that $|E_1-E_2|$ is sufficiently small.
\end{theorem}

The rest of this paper is organized as follows. Except for some statements in applications (Section 3), this paper is entirely self contained.
 We will introduce many applications to quasi-periodic operators   in Section 3.
 Sections  4, 5, 6, 7  are devoted to  prove Theorems \ref{thmmul}, \ref{thmmu2},  \ref{thmu},  \ref{thmunew} and \ref{thmholder}.
 We  will introduce the discrepancy  for  semi-algebraic sets  in Section 8. In Section 9, we will give  the proof for  all   the results in Section 3.
 
\section{Applications}\label{sapp}

Let $S$ be a Toeplitz  (operator) matrix on $\ell^2(\Z^d)$ with respect to $f$, namely,
\begin{equation}\label{gdec72s}
  S(x;n+k,n^\prime+k)= S(f^k(x);n,n^\prime),
\end{equation}
for any $n\in \Z^d, n^\prime\in \Z^d$ and $k\in \Z^d.$
Assume every element $S(z;n,n^\prime)$, $n,n^\prime \in\Z^d$,   is analytic in a strip $\{z:|\Im z|\leq \rho\}$ with $\rho>0$ and  satisfies for any $x\in\R$ and $n,n^\prime\in \Z^d$, 
\begin{equation}\label{GOnewnews}
|S(x;n,n')|\leq  K e^{-c_1|n-n'|},  K>0, c_1>0.
\end{equation}
Assume that  there exists $K_1>1$ such that for any $x\in\T^b$ and $z\in \{z\in\C^b:|\Im z|\leq \rho\}$ with $||x-z||\leq e^{ -\left(\log  (|n|+|n^\prime|+2)\right)^{K_1}}$, 
\begin{equation}\label{glc1news}
|A(x;n,n')-A(z;n,n')|\leq  K ||x-z||^\gamma.
\end{equation}

Assume for any $N>1$, $n,n^\prime\in\Z^d$ with $|n|\leq N$ and $|n^\prime|\leq N$, there exists a trigonometric polynomial  $\tilde{S}(x;n,n^\prime)$ of degree less than $ e^{(\log N)^{K_1}}$ such that 
\begin{equation}\label{glc1news1}
\sup_{x\in\T^b} |S(x;n,n')-\tilde{S}(x;n,n^\prime)|\leq  Ke^{-N^{2}}.
\end{equation}

Define a family of  operators $H(x)$  on $\ell^2(\Z^d)$:
\begin{equation}\label{ops}
  H(x)=\lambda^{-1}S+v(f(n,x))\delta_{nn^\prime},
\end{equation}
where  $v$ is an  analytic   function on $\T^b$.

 In this section,   we always assume 
 \begin{itemize}
 	\item $v$ is non-constant,
 \item 	$f$ is   a   frequency shift or skew-shift,
 \item except for subsection \ref{Sinter}, $S$ is   a Toeplitz  (operator) matrix on $\ell^2(\Z^d)$ with respect to $f$ and satisfies   \eqref{gdec72s}-\eqref{glc1news1}.
 \end{itemize}

 {\bf Example 4}:   
\begin{itemize}
	\item If $S$ is a long range operator, namely,  $S$ does not depend on $x$ and 
	$$S(n,n^\prime)\leq K e^{-c_1|n-n'|},n,n^\prime\in\Z^d,$$
	then  \eqref{GOnewnews}, \eqref{glc1news} and \eqref{glc1news1}  hold.
	\item Let $\phi_k(x)$, $k\in\Z$,  be a  trigonometric polynomial  on $\T^b$ of degree less than $e^{(\log (1+|k|))^{K_1}}$satsifying 
	\begin{equation*}
	\sup_{x\in\T^b}|\phi_k(x)|\leq K e^{-c_1|k|}.
	\end{equation*}
	Let
	\begin{equation*}
	S(x;n,n^\prime)=\phi_{n-n^\prime} (f(n,x)) +\overline{\phi_{n^\prime-n}(f(n^\prime, x))}.
	\end{equation*}
	Then \eqref{GOnewnews}, \eqref{glc1news} and \eqref{glc1news1}  hold. 
\end{itemize}

 \begin{remark}\label{rebound}
 	For $d\geq 2$, our settings \eqref{gdec72s}-\eqref{glc1news1} is the first time to allow every  entry of $S$ to depend on $x$, which beyonds  the long range operators. 
 	For $d=1$,  Bourgain \cite{bkick} studied the case in Example 4  under the assumption that   $\phi_k(x)$ is  a  trigonometric polynomial   of  degree  at most $N^C$. 
 \end{remark}

 We will apply Theorems \ref{thmmul}, \ref{thmmu2},  \ref{thmu} and \ref{thmholder}  to  operators  $$A(x)=H(x)=\lambda^{-1}S+v(f(n,x))\delta_{nn^\prime}.$$ In this section, the Green's functions always depend on energy $E$.
 See \eqref{ge}.
 
The IDS appearing in applications is always  existed, namely, the following limit
 \begin{equation*}
 k(x, E)=\lim_{N\to\infty}\frac{1}{(2N+1)^d}\#\{\text{ eigenvalues of } R_{{[-N,N]^d}}A(x)R_{{[-N,N]^d}} \text{ smaller than } E\},
 \end{equation*}
 converges to $k(E)$ for almost every $x$.   We  write  $k(E)$ for the IDS  when it exists.
 
 For the large deviation theorem,  $S$ is not necessarily self-adjoint. However, in order to  establish
 pure point spectrum property,  self-adjointness is necessary because of the energy elimination.

\subsection{Shifts: $d=1$, arbitrary $b$}
Denote by $\Delta$ the discrete Laplacian on $\ell^2(\Z)$, that is, for $\{u(n)\}\in \ell^2(\Z)$,
\begin{equation*}
(\Delta u)(n)=\sum_{|n-n^\prime|=1} u(n^\prime).
\end{equation*}
We say that $\omega=(\omega_1,\omega_2,\cdots,\omega_b)$ satisfies Diophantine condition ${\rm DC}(\kappa,\tau)$, if
\begin{equation}\label{gdc}
||k\omega||\geq \frac{\tau}{|k|^{\kappa}},k\in \Z^{b}\backslash \{(0,0,\cdots,0)\}.
\end{equation}
By the Dirichlet principle, one has $\kappa\geq b$.
When $\kappa>b$, $\cup_{\tau>0}{\rm DC}(\kappa,\tau)$ has full Lebesgue measure.

We say that $\omega\in \R$ satisfies strong Diophantine conditions   if there exist $\kappa>1$ and $\tau>0$ such that
\begin{equation}\label{gsdc}
||k\omega||\geq \frac{\tau}{k(1+\log k)^\kappa} \text{ for all } k\in \N.
\end{equation}
It is easy to see that almost every $\omega $ satisfies strong Diophantine conditions.

Let
\begin{equation*}
  f^n(x)=x+n\omega=(x_1+n\omega_1,x_2+n\omega_2,\cdots,x_b+ n\omega_b)\mod\Z^b,
\end{equation*}
where $x=(x_1,x_2,\cdots,x_b)\in \T^b$, $n\in \Z$ and $\omega=(\omega_1,\omega_2,\cdots,\omega_b)\in \R^b$.

Let $H(x)$ on $\ell^2(\Z)$ be given by
\begin{equation}\label{opapp1}
  H(x)=\Delta+ v(f^n(x))=\Delta+ v(x_1+n\omega_1,x_2+n\omega_2,\cdots,x_b+n\omega_b)\delta_{nn'},
\end{equation}
where $n,n^\prime\in \Z$.

  Let
\begin{equation}\label{G.transfer}
A_{k}^E(x)=\prod_{j=k-1}^{0 }A^E(x+j\omega)=A^E(x+(k-1)\omega)A^E(x+(k-2)\omega)\cdots A^E(x)
\end{equation}
and
\begin{equation}\label{G.transfer1}
A_{-k}^E(x)=(A_{k}^{E}(x-k\omega))^{-1}
\end{equation}
for $k\geq 1$,
where $A^E(x)=\left(
             \begin{array}{cc}
               E- v(x) & -1 \\
               1& 0\\
             \end{array}
           \right)
$.
$A_{k}^E$  is called the (k-step) transfer matrix.
  The Lyapunov exponent
is given  by
 \begin{equation}\label{G21}
    L(E)=\lim_{k\rightarrow\infty} \frac{1}{k}\int_{\T^b} \ln \| A_k^E(x)\|dx.
 \end{equation}

\begin{theorem}\label{thmapp1}
	Let  $\omega\in {\rm DC}(\kappa,\tau)$ and $1-\frac{1}{b\kappa}<\sigma<1$. Let  $H(x)$ be given by \eqref{opapp1}.  
Assume the Lyapunov exponent  $L(E)$ is positive.
Then for any $\varepsilon>0$
and large $N$, there exists a subset $X_N\subset \T^b$ such that
\begin{equation*}
 {\rm Leb} (X_N)\leq e^{-N^{\frac{\sigma-1}{b^2\kappa}+\frac{1}{b^3\kappa^2}-\varepsilon}},
\end{equation*}
and for any $x  \notin X_N$, we have
\begin{equation*}
  ||  G_{[-N,N]}(E,x)||\leq e^{N^{\sigma}},
\end{equation*}
and
\begin{equation*}
  |G_{[N,-N]}(E,x;n,n^\prime)|\leq  e^{-(L(E)-\varepsilon)|n-n^\prime|} \text{ for } |n-n^\prime|\geq N/10.
\end{equation*}
\end{theorem}
\begin{theorem}\label{thmapp1'}
Let $\omega\in {\rm DC}(\kappa,\tau)$ and $H(x)$ be given by \eqref{opapp1}.
Suppose the Lyapunov exponent  $L(E)>0$ for every  $E$ in  an interval $I$.
Then for any $\varepsilon>0$,  
\begin{equation*}
  |k(E_1)-k(E_2)|\leq  e^{-\left(\log\frac{1}{ |E_1-E_2|}\right)^{\frac{1}{b^3\kappa^2}-\varepsilon}}
\end{equation*}
provided that $|E_1-E_2|$ is sufficiently small and $E_1,E_2\in I$.
\end{theorem}

\begin{theorem}\label{coroapp1}
Let $H(x)$ be given by \eqref{opapp1}.  Then the following statement is true for almost every $\omega$.
Assume  the Lyapunov exponent  $L(E)>0$ for every $E$ in  an interval $I$. Then for any $\varepsilon>0$,
\begin{equation*}
  |k(E_1)-k(E_2)|\leq  e^{-\left(\log\frac{1}{ |E_1-E_2|}\right)^{\frac{1}{b^3}-\varepsilon}},
\end{equation*}
provided that $|E_1-E_2|$ is sufficiently small and  $E_1,E_2\in I$.
\end{theorem}
\begin{remark}\label{reb}
	Under the same assumptions, the large deviation theorem and
	the modulus of  continuity  of the IDS were shown in  \cite{gs01} (also see \cite{bbook}).  When $b=2$,   a better bound $b=1/3$ was obtained in \cite{gs01}.
	However, there are no explicit bounds in \cite{gs01,bbook} when $b\geq 3$.
\end{remark}
Putting  a coupling constant $\lambda^{-1}$ in front of the Laplacian $\Delta$, the operator given by \eqref{opapp1}  becomes
\begin{equation}\label{opapp1new}
  H(x)=\lambda^{-1}\Delta+ v(x+n\omega)\delta_{nn^\prime}.
\end{equation}
For large $\lambda$ only depending on the potential $v$,  the Lyapunov exponent $L(E)$ is  positive for every $E$ \cite{bjam}.
Therefore, we have the  following three corollaries
\begin{corollary}\label{coroapp10}
Assume $\omega\in {\rm DC}(\kappa,\tau)$ and $1-\frac{1}{b\kappa}<\sigma<1$.
Let $H(x)$ be given by \eqref{opapp1new}.
Then there exists $\lambda_0=\lambda_0(v)$ such that for any $\varepsilon>0$, $\lambda>\lambda_0$
and large $N$, there exists  $X_N\subset \T^b$ such that
\begin{equation}\label{gdec89}
 {\rm Leb} (X_N)\leq e^{-N^{\frac{\sigma-1}{b^2\kappa}+\frac{1}{b^3\kappa^2}-\varepsilon}},
\end{equation}
and for any $x  \notin X_N$, we have
\begin{equation*}
  ||  G_{[-N,N]}(E,x)||\leq e^{N^{\sigma}},
\end{equation*}
and
\begin{equation*}
  |G_{[N,-N]}(E,x;n,n^\prime)|\leq  e^{-(L(E)-\varepsilon)|n-n^\prime|} \text{ for } |n-n^\prime|\geq N/10.
\end{equation*}
\end{corollary}
\begin{corollary}\label{thmapp1'new}
Let $\omega\in {\rm DC}(\kappa,\tau)$  and $H(x)$ be given by \eqref{opapp1new}.
Then there exists $\lambda_0=\lambda_0(v)$ such that for any $\varepsilon>0$ and $\lambda>\lambda_0$,  
\begin{equation*}
  |k(E_1)-k(E_2)|\leq  e^{-\left(\log\frac{1}{ |E_1-E_2|}\right)^{\frac{1}{b^3\kappa^2}-\varepsilon}}
\end{equation*}
provided that $|E_1-E_2|$ is sufficiently small.
\end{corollary}
\begin{corollary}\label{coroapp1new}
Let $H(x)$ be given by \eqref{opapp1new}.
 Then there exists $\lambda_0=\lambda_0(v)$ such that   the following statement is true for   almost every $\omega$.  For any $\varepsilon>0$ and $\lambda>\lambda_0$, 
\begin{equation*}
  |k(E_1)-k(E_2)|\leq  e^{-\left(\log\frac{1}{ |E_1-E_2|}\right)^{\frac{1}{b^3}-\varepsilon}},
\end{equation*}
provided that $|E_1-E_2|$ is sufficiently small.
\end{corollary}

Let $H(x)$ on $\ell^2(\Z)$ be given by
\begin{equation}\label{opapp2}
  H(x)=\lambda^{-1}S+ v(f^n(x))=\lambda^{-1}S+ v(x+n\omega)\delta_{nn^\prime},
\end{equation}
where $x,\omega\in \R^b$.

\begin{theorem}\label{thmapp4}
Let $H(x)$ be given by \eqref{opapp2}.
Assume  $\omega\in {\rm DC}(\kappa,\tau)$ and $1-\frac{1}{b\kappa}<\sigma<1$. Then for any $\varepsilon>0$, there exists $$\lambda_0=\lambda_0(\varepsilon,\kappa,\tau,\rho,\sigma,\gamma,K,K_1,c_1,v)$$ such that
for any $\lambda>\lambda_0$ and  any $N$,  there exists $X_N\subset \T^b$ such that
\begin{equation*}
 {\rm Leb} (X_N)\leq e^{-N^{\frac{\sigma-1}{b^2\kappa}+\frac{1}{b^3\kappa^2}-\varepsilon}},
\end{equation*}
and for any $x\notin X_N$, we have
\begin{equation*}
  ||  G_{[-N,N]}(E,x)||\leq e^{N^{\sigma}},
\end{equation*}
and
\begin{equation*}
  |G_{[N,-N]}(E,x;n,n^\prime)|\leq  e^{-\frac{1}{2}c_1|n-n^\prime|} \text{ for } |n-n^\prime|\geq N/10.
\end{equation*}
\end{theorem}
\begin{theorem}\label{thmapp4'}
		Assume $S$ is self-adjoint and $\omega\in {\rm DC}(\kappa,\tau)$.  Let $H(x)$ be given by \eqref{opapp2}. Then for any $\varepsilon>0$, there exists $$\lambda_0=\lambda_0(\varepsilon,\kappa,\tau,\rho,\gamma,K,K_1,c_1,v)$$ such that
for any $\lambda>\lambda_0$,
\begin{equation*}
  |k(E_1)-k(E_2)|\leq  e^{-\left(\log\frac{1}{ |E_1-E_2|}\right)^{\frac{1}{b^3\kappa^2}-\varepsilon}},
\end{equation*}
provided that $|E_1-E_2|$ is sufficiently small.
\end{theorem}
\begin{theorem}\label{coroapp4}
		Assume $S$ is self-adjoint.
Let $H(x)$ be given by \eqref{opapp2}.
 Then  for almost every $\omega\in \R^b$ the following is true. For  any $\varepsilon>0$, there exists $$\lambda_0=\lambda_0(\varepsilon,\omega,\rho,\gamma,K,K_1,c_1,v)$$ such that
for any $\lambda>\lambda_0$,
\begin{equation*}
  |k(E_1)-k(E_2)|\leq  e^{-\left(\log\frac{1}{ |E_1-E_2|}\right)^{\frac{1}{b^3}-\varepsilon}},
\end{equation*}
provided that $|E_1-E_2|$ is sufficiently small.
\end{theorem}
\begin{theorem}\label{thmapp4new}
	Let   $H(x)$ be given by \eqref{opapp2}. Then for any $\varrho>0$, there is  $\lambda_0=\lambda_0(\varrho,\rho,\gamma,K,K_1,c_1,v)>0$ such that the following statement holds.
	For any $\lambda> \lambda_0$ and any $x\in\mathbb{T} $, there exists  $\Omega=\Omega(x,\lambda, S,v,\varrho)\subset \mathbb{T}^{b}$ with $\mathrm{Leb} (\mathbb{T}^{b}\setminus \Omega)\leq \varrho $ such that for any $\omega\in \Omega$, $H(x)$ satisfies Anderson localization.
\end{theorem}
\subsection{Shifts: $b=1$, arbitrary $d$}
Let $v $ be analytic on  $\T$.
Let
\begin{equation*}
  f^n(x)=x+n\omega=x+n_1\omega_1+n_2\omega_2+\cdots +n_d\omega_d\mod\Z,
\end{equation*}
where $n=(n_1,n_2,\cdots,n_d)\in \Z^d$ and $x\in \T$.
Let $H(x)$ on $\ell^2(\Z^d)$ be given by
\begin{equation}\label{opapp4}
  H(x)=\lambda^{-1}S+ v(f^n(x))\delta_{nn^\prime}=\lambda^{-1}S+ v(x+n_1\omega_1+n_2\omega_2+\cdots +n_d\omega_d)\delta_{nn^\prime}.
\end{equation}

\begin{theorem}\label{thmapp2}
Let $\omega\in {\rm DC}(\kappa,\tau)$ and  $H(x)$ be given by \eqref{opapp4}. Then for any $\varepsilon>0$, there exists $\lambda_0=\lambda_0(\varepsilon,\kappa,\tau,\rho,\sigma,\gamma,K,K_1,c_1,v)$ such that
for any $\lambda>\lambda_0$ and any $N$, there exists $X_N\subset \T$ such that
\begin{equation}\label{gdec2}
 {\rm Leb} (X_N)\leq e^{-N^{\sigma-\varepsilon}},
\end{equation}
and for any $x\notin X_N$ and any $Q_N\in \mathcal{E}_N^0$, we have
\begin{equation*}
  ||  G_{Q_N}(E,x)||\leq e^{N^{\sigma}},
\end{equation*}
and
\begin{equation}\label{gdec1}
  |G_{Q_N}(E,x;n,n^\prime)|\leq  e^{-\frac{1}{2}c_1|n-n^\prime|} \text{ for } |n-n^\prime|\geq N/10.
\end{equation}
\end{theorem}
\begin{theorem}\label{thmapp2'}
		Assume $S$ is self-adjoint and $\omega\in {\rm DC}(\kappa,\tau)$.
Let $H(x)$ be given by \eqref{opapp4}.
 Then for any $\varepsilon>0$, there exists $$\lambda_0=\lambda_0(\varepsilon,\kappa,\tau,\rho,\gamma,K,K_1,c_1,v)$$ such that
for any $\lambda>\lambda_0$,
\begin{equation*}
  |k(E_1)-k(E_2)|\leq  e^{-\left(\log\frac{1}{ |E_1-E_2|}\right)^{1-\varepsilon}},
\end{equation*}
provided that $|E_1-E_2|$ is sufficiently small.
\end{theorem}
\begin{theorem}\label{thmapp2''}
		Assume $S$ is self-adjoint.
Let $H(x)$ be given by \eqref{opapp4}.
Then for any $\varrho>0$, there is  $\lambda_0=\lambda_0(\varrho,\rho,\gamma,K,K_1,c_1,v)>0$ such that the following statement holds.
For any $\lambda> \lambda_0$ and any $x\in\mathbb{T} $, there exists  $\Omega=\Omega(x,\lambda, S,v,\varrho)\subset \mathbb{T}^{d}$ with $\mathrm{Leb} (\mathbb{T}^{d}\setminus \Omega)\leq \varrho $ such that for any $\omega\in \Omega$, $H(x)$ satisfies Anderson localization.
\end{theorem}
\begin{remark}
Theorem \ref{thmapp2''} is a generalization of Theorem 2 in p.138 of  \cite{bbook} and main result in \cite{cd93}.
\end{remark}
\subsection{ Skew-shifts: $d=1$, arbitrary $b$}
Let $f$: $\T^b\rightarrow \T^b$ be the skew-shift   defined as follows
\begin{equation}\label{gdec5skew}
f(x_1,x_2,...,x_b)=(x_1+\omega, x_2+x_1,...,x_b+x_{b-1}).
\end{equation}

Let $H(x)$ on $\ell^2(\Z)$ be given by
\begin{equation}\label{opapp5}
  H(x)=\lambda^{-1}S(x)+ v(f^{n}(x))\delta_{nn^\prime}, 
\end{equation}
$v $ is analytic on  $\T^b$.
\begin{theorem}\label{thmapp5}
		
Let $H(x)$ be given by \eqref{opapp5}.
Assume   $\omega\in {\rm DC}(\kappa,\tau)$ and $ 1-\frac{1}{2^{b-1}b\kappa}<\sigma<1$. Then for any $\varepsilon>0$, there exists $\lambda_0=\lambda_0(\varepsilon,\kappa,\tau,\rho,\sigma,\gamma,K,K_1,c_1,v)$ such that
for any $\lambda>\lambda_0$ and  any $N$, there exists $X_N\subset \T^b$ such that
\begin{equation*}
 {\rm Leb} (X_N)\leq e^{-N^{\frac{\sigma-1}{2^{b-1}b^2\kappa}+\frac{1}{4^{b-1}b^3\kappa^2}-\varepsilon}},
\end{equation*}
and for any $x\notin X_N$, we have
\begin{equation*}
  ||  G_{[-N,N]}(E,x)||\leq e^{N^{\sigma}}
\end{equation*}
and
\begin{equation*}
  |G_{[N,-N]}(E,x;n,n^\prime)|\leq  e^{-\frac{1}{2}c_1|n-n^\prime|} \text{ for } |n-n^\prime|\geq N/10.
\end{equation*}
\end{theorem}
\begin{remark}
Under stronger assumptions that $\omega\in {\rm DC}(2,\tau)$, $v$ and  each element of $S$   are   nonconstant trigonometric polynomials,
the large deviation theorem appearing in Theorem \ref{thmapp5} without explicit bounds was proved   for $d=2$ \cite{bbook} and arbitrary $d$ \cite{sy19}.
\end{remark}
\begin{theorem}\label{thmapp5'}
		Assume $S$ is self-adjoint  and $\omega\in {\rm DC}(\kappa,\tau)$.
Let $H(x)$ be given by \eqref{opapp5}.
 Then for any $\varepsilon>0$, there exists $$\lambda_0=\lambda_0(\varepsilon,\kappa,\tau,\rho,\gamma,K,K_1,c_1,v)$$ such that
for any $\lambda>\lambda_0$, we have
\begin{equation*}
  |k(E_1)-k(E_2)|\leq  e^{-\left(\log\frac{1}{ |E_1-E_2|}\right)^{\frac{1}{4^{b-1}b^3\kappa^2}-\varepsilon}},
\end{equation*}
provided that $|E_1-E_2|$ is sufficiently small.
\end{theorem}

\begin{corollary}\label{coroapp5}
		Assume $S$ is self-adjoint.
Let $H(x)$ be given by \eqref{opapp5}.
 Then  for almost every $\omega\in \R$ the following is true. For  any $\varepsilon>0$, there exists $\lambda_0=\lambda_0(\varepsilon,\omega,\rho,\gamma,K,K_1,c_1,v)$ such that
for any $\lambda>\lambda_0$,
\begin{equation*}
  |k(E_1)-k(E_2)|\leq  e^{-\left(\log\frac{1}{ |E_1-E_2|}\right)^{\frac{1}{4^{b-1}b^3}-\varepsilon}},
\end{equation*}
provided that $|E_1-E_2|$ is sufficiently small.
\end{corollary}
Assume $S$ is taken the particular case, i.e., $S=\Delta.$  Let $b=2$. In this case, by Corollary \ref{coroapp5},
$\frac{1}{4^{b-1}b^3}=\frac{1}{32}$.  A  bound $\frac{1}{24}$ was shown by
 Bourgain, Goldstein and Schlag \cite{bgscmp}.
By  combining the arguments in Bourgain, Goldstein and Schlag \cite{bgscmp} with  the proof of Corollary \ref{coroapp5},  we are able to  improve the bound.
\begin{corollary}\label{coroapp5b}
		Assume $S$ is self-adjoint.
Let $b=2$
 and $H(x)$ be given by \eqref{opapp5}. 
 Then  for almost every $\omega\in \R$ the following is true. For  any $\varepsilon>0$, there exists $\lambda_0=\lambda_0(\varepsilon,\omega,\rho,\gamma,K,K_1,c_1,v)$ such that
for any $\lambda>\lambda_0$,
\begin{equation*}
  |k(E_1)-k(E_2)|\leq  e^{-\left(\log\frac{1}{ |E_1-E_2|}\right)^{\frac{1}{18}-\varepsilon}},
\end{equation*}
provided that $|E_1-E_2|$ is sufficiently small.
\end{corollary}

\subsection{Skew-shifts: $d=b=1$}
Let $P_b$ be the projection on the $b$th coordinate of $\T^b$, namely, $P_b(x_1,x_2,\cdots,x_b)=x_b$, where $(x_1,x_2,\cdots,x_b)\in\R^b$.
Define $H(x)$ on $\ell^2(\Z)$,
 \begin{equation}\label{opapp7}
  H(x)=\lambda^{-1}\Delta+ v(P_{b}(f^{n}(x)))\delta_{nn^\prime},
\end{equation}
where $v$ is analytic on $\T$ and $f$ is the skew-shift on $\T^b$.
\begin{theorem}\label{thmapp7}
Let $H(x)$ be given by \eqref{opapp7}.
Assume    $\omega$ is strong Diophantine and $ 1-\frac{1}{2^{b-1}b}<\sigma<1$.
Then there exists $\lambda_0=\lambda_0(v)$ such that for any $\varepsilon>0$, $\lambda>\lambda_0$
and large $N$, there exists   $X_N\subset \T^b$ such that
\begin{equation*}
 {\rm Leb} (X_N)\leq e^{-N^{\frac{\sigma-1}{2^{b-1}}+\frac{1}{4^{b-1}}-\varepsilon}},
\end{equation*}
and for any $x\notin X_N$, we have
\begin{equation*}
  ||  G_{[-N,N]}(E,x)||\leq e^{N^{\sigma}}
\end{equation*}
and
\begin{equation*}
  |G_{[N,-N]}(E,x;n,n^\prime)|\leq  e^{-\frac{1}{2}c_1|n-n^\prime|} \text{ for } |n-n^\prime|\geq N/10.
\end{equation*}
\end{theorem}

\begin{theorem}\label{thmapp7'}
Let $\omega$ be strong Diophantine  and $H(x)$ be given by \eqref{opapp7}.
Then there exists $\lambda_0=\lambda_0(v)$ such that for any $\varepsilon>0$ and $\lambda>\lambda_0$,  
\begin{equation}\label{gdec61}
  |k(E_1)-k(E_2)|\leq  e^{-\left(\log\frac{1}{ |E_1-E_2|}\right)^{\frac{1}{4^{b-1}}-\varepsilon}},
\end{equation}
provided that $|E_1-E_2|$ is sufficiently small.
\end{theorem}
\begin{remark}
\begin{itemize}
\item Comparing to  Theorems \ref{thmapp5} and \ref{thmapp5'}, there is no dimension ($b^3$) loss in  the bounds of  Theorems \ref{thmapp7} and \ref{thmapp7'}. This is because  the potential $v$ is defined on $\T$.
  \item The large deviation theorem and  the modulus of continuity of  Lyapunov exponents (the IDS) without explicit bounds was  obtained in \cite{taojde}.
   \item Let $b=2$. The constant in \eqref{gdec61} becomes $ \frac{1}{4^{b-1}}=\frac{1}{4}$. It is possible to improve the bound from  $1/4$ to $1/3$ by incorporating the arguments in  \cite{bgscmp}.  A weaker result was proved  by Tao \cite{taoejde}, where a constant $\frac{1}{30}$ was obtained. 

\end{itemize}

\end{remark}

\subsection{Shifts: $d=b=2$}
Assume $v$ is analytic on $\T^2=(\R/\Z)^2$.
Let
\begin{equation*}
  f^n(x)=(x_1+n_1\omega_1,x_2+n_2\omega_2)\mod\Z^2,
\end{equation*}
where $n=(n_1,n_2)\in \Z^2$, $\omega=(\omega_1,\omega_2)\in \R^2$ and $x=(x_1,x_2)\in \T^2$.
Let $H(x)$ on $\ell^2(\Z^2)$ be given by
\begin{equation}\label{opapp6}
  H(x)=\lambda^{-1}S(x)+v(f^n(x))\delta_{nn^\prime}=\lambda^{-1}S(x_1,x_2)+
   v(x_1+n_1\omega_1,x_2+n_2\omega_2)\delta_{nn^\prime}.
\end{equation}

\begin{theorem}\label{thmapp3}
Let $H(x)$ be given by \eqref{opapp6}.
Suppose $v$ is nonconstant on any line segment contained in $[0,1)^2$, $\omega_1\in {\rm DC}(\kappa,\tau)$ and $\omega_2\in {\rm DC}(\kappa,\tau)$ with $1\leq  \kappa<\frac{13}{12}$. Assume $$3\kappa-\frac{9}{4}<\sigma<1.$$
Then there exists $\lambda_0=\lambda_0(\varepsilon,\kappa,\tau,\rho,\sigma,\gamma,K,K_1,c_1,v)$ such that
for any $\lambda>\lambda_0$ and any $N$,  there exists $X_N\subset \T^2$ such that for any line segment $L\subset [0,1)^2$,
\begin{equation}\label{gdec33}
 {\rm Leb} (X_N\cap L)\leq e^{-N^{(\sigma-1)(13/4-3\kappa)+ (13/4-3\kappa)^2-\varepsilon}},
\end{equation}
and   for any $x\notin X_N$ and $Q_N\in \mathcal{E}_N^0$, we have
\begin{equation*}
  ||  G_{Q_N}(E,x)||\leq e^{N^{\sigma}}
\end{equation*}
and
\begin{equation*}
  |G_{Q_N}(E,x;n,n^\prime)|\leq  e^{-\frac{1}{2}c_1|n-n^\prime|} \text{ for } |n-n^\prime|\geq N/10.
\end{equation*}
\end{theorem}
\begin{theorem}\label{thmapp3'}
		Assume $S$ is self-adjoint, $v$ is nonconstant on any line segment contained in $[0,1)^2$, $\omega_1\in {\rm DC}(\kappa,\tau)$ and $\omega_2\in {\rm DC}(\kappa,\tau)$ with $1\leq  \kappa<\frac{13}{12}$.  
Let $H(x)$ be given by \eqref{opapp6}. Then for any $\varepsilon$, there exists $\lambda_0=\lambda_0(\varepsilon,\kappa,\tau,\rho,\gamma,K,K_1,c_1,v)$ such that
for any $\lambda>\lambda_0$,
\begin{equation*}
  |k(E_1)-k(E_2)|\leq  e^{-\left(\log\frac{1}{ |E_1-E_2|}\right)^{(13/4-3\kappa)^2-\varepsilon}},
\end{equation*}
provided that $|E_1-E_2|$ is sufficiently small.
\end{theorem}
\begin{corollary}\label{coroapp3}
		Assume $S$ is self-adjoint and $v$ is nonconstant on any line segments contained in $[0,1)^2$.
Let $H(x)$ be given by \eqref{opapp6}.
Then  for almost every $\omega\in \R^2$ the following is true.  For  any $\varepsilon>0$, there exists $\lambda_0=\lambda_0(\varepsilon,\omega,\rho,\gamma,K,K_1,c_1,v)$ such that
for any $\lambda>\lambda_0$,
\begin{equation*}
  |k(E_1)-k(E_2)|\leq  e^{-\left(\log\frac{1}{ |E_1-E_2|}\right)^{\frac{1}{16}-\varepsilon}},
\end{equation*}
provided that $|E_1-E_2|$ is sufficiently small.
\end{corollary}

\begin{remark}
   Theorems  \ref{thmapp3} and \ref{thmapp3'}  follow the arguments  in \cite{bk}. Our quantitative approaches developed in the paper allow us to obtain the explicit bounds.

\end{remark}

 \subsection{Sub-exponentially decaying matrices with interactions}\label{Sinter}
 Our applications  can be  wider.   
 Here are several examples.
  Instead of \eqref{GOnewnews}, assume   
 \begin{equation}\label{GOnewnewsnew}
 |S(x;n,n')|\leq  K e^{-c_1|n-n'|^{\tilde{\sigma}}}, 0<\tilde{  \sigma}\leq1,c_1>0,
 \end{equation}
 for any $n,n^\prime\in \Z^d$.
  
 	Assume for any $N>1$, $n,n^\prime\in\Z^d$ with $|n|\leq N$ and $|n^\prime|\leq N$, there exists a trigonometric polynomial  $\tilde{S}(x;n,n^\prime)$ of degree less than $ e^{N^a}$ such that 
 	\begin{equation}\label{glc1news2}
 	\sup_{x\in\T^b} |S(x;n,n')-\tilde{S}(x;n,n^\prime)|\leq  	Ke^{-N^{2}}.
 	\end{equation}

 In this subsection, assume $S$ satisfies  \eqref{glc1news}, \eqref{GOnewnewsnew} and \eqref{glc1news2}.
 
 Let $\tilde{U}$ be a diagonal matrix on  $\ell^2(\Z^d) $ satisfying 
 $$||U||\leq K. $$
 Given $m\in \Z^d$, define the  diagonal matrix  $\tilde{U}^{m}$ on $\ell^2(\Z^d)$ by 
 $$\tilde{U}^{m}(n)=\tilde{U}(m+n),n\in\Z^d.$$
 We say $\tilde{U}$ has low complexity if  there exists $0<a<1$ such that  for any $N>1$,
 \begin{equation}\label{comu1}
 \#\{ R_{Q_N} U^m(n)\delta_{nn^\prime} R_{Q_N}: m\in\Z^d,Q_{N}\in \mathcal{E}_{N}^0 \}\leq Ke^{N^a}.
 \end{equation}
 
%
 
 Let 
 \begin{equation}\label{newh}
\tilde{H}(x)=H(x)+\lambda^{-1}U+\tilde{U}=\lambda^{-1}(S+U)+(\tilde{U}(n) +v(f(n,x)))\delta_{nn^\prime}.
 \end{equation}
 For any $m\in \Z^d$,  
 let 
 \begin{equation}\label{newhm}
 \tilde{H}^m(x)=H(x)+\lambda^{-1}U^m+\tilde{U}^m=\lambda^{-1}(S+U^m)+(\tilde{U}^m+v(f(n,x)))\delta_{nn^\prime}.
 \end{equation}
 Denote by $\tilde{G}^m$ the Green's function of $\tilde{H}^m$.
%

\begin{theorem}\label{ldtfeb2new}
	Assume   $\alpha$ is strong Diophantine, and   $U $  and $\tilde{U}$ have low complexity in the sense of \eqref{comu} and \eqref{comu1} repectively.
	 Assume
$$1-\frac{1}{b}<\sigma<\tilde{\sigma} \text{ and } a\leq \frac{1}{4} \left\{\frac{1}{K_1}, \frac{\sigma-1}{b^2}+\frac{1}{b^3}\right\}.$$ Let $H(x)$ and $\tilde{H}^m(x)$ be given by \eqref{opapp2}   and  \eqref{newhm} respectively. 
Then for any $\varepsilon>0$, there exists $$\lambda_0=\lambda_0(\varepsilon,\alpha,\rho,c_1,\sigma,\tilde{\sigma}, \gamma,K,K_1,c_1,v)$$ such that
	for any $\lambda>\lambda_0$ and  any $N$,  there exists $X_N\subset \T^b$ such that
	\begin{equation*}
	{\rm Leb} (X_N)\leq e^{-N^{\frac{\sigma-1}{b^2 }+\frac{1}{b^3}-\varepsilon}},
	\end{equation*}
	and for  any  $x\notin X_N$ and  $m\in\Z$, we have
	\begin{equation*}
	||  \tilde{G}_{[-N,N]}^m(E,x)||\leq e^{N^{\sigma}},
	\end{equation*}
	and
	\begin{equation*}
	|\tilde{G}_{[-N,N]}^m(E,x;n,n^\prime)|\leq  e^{-\frac{c}{2}|n-n^\prime|} \text{ for } |n-n^\prime|\geq N/10,
	\end{equation*}
	where $c=\frac{5^{\tilde{\sigma}}-1}{5^{\tilde{\sigma}}}$.
\end{theorem}

\begin{theorem}\label{thmapp2new}

		Assume    $\omega\in {\rm DC}(\kappa,\tau)$, $U $   and $\tilde{U}$ have low complexity.
Assume  $0<\sigma<\tilde{\sigma}\leq 1$ and  $a\leq \frac{1}{4}\min\{\frac{1}{K_1}, \sigma\}$.
	Let $H(x)$  and $\tilde{H}^m(x)$ be given by  \eqref{opapp4}  and  \eqref{newhm} respectively.
 Then for any $\varepsilon>0$, there exists $$\lambda_0=\lambda_0(\varepsilon,\kappa,\tau,\sigma,\tilde{\sigma},\rho,\gamma,K,K_1,c_1,v)$$ such that
for any $\lambda>\lambda_0$ and any $N$, there exists $X_N\subset \T$ such that
\begin{equation}\label{gdec2new}
 {\rm Leb} (X_N)\leq e^{-N^{\sigma-\varepsilon}},
\end{equation}
and for any $x\notin X_N$, any $m\in \Z^d$ and  any $Q_N\in \mathcal{E}_N^0$, we have
\begin{equation*}
  ||  \tilde{G}^m_{Q_N}(E,x)||\leq e^{N^{\sigma}},
\end{equation*}
and
\begin{equation}\label{gdec1new}
  |\tilde{G}^m_{Q_N}(E,x;n,n^\prime)|\leq  e^{-\frac{c}{2}|n-n^\prime|^{\tilde{\sigma}} }\text{ for } |n-n^\prime|\geq N/10,
\end{equation}
where $c=\frac{5^{\tilde{\sigma}}-1}{5^{\tilde{\sigma}}}$.
\end{theorem}
\begin{theorem}\label{thmapp2'new}
		Assume $S$  is self-adjoint, $\omega\in {\rm DC}(\kappa,\tau)$ and  $a\leq\frac{1}{4}\min\{\frac{1}{K_1},\tilde{\sigma}\}$.
	Let $H(x)$ be given by  \eqref{opapp4}.
  Then for any $\varepsilon>0$, there exists $$\lambda_0=\lambda_0(\varepsilon,\kappa,\tau,\tilde{\sigma},\rho,\gamma,K,K_1,c_1,v)$$ such that
for any $\lambda>\lambda_0$,
\begin{equation*}
  |k(E_1)-k(E_2)|\leq  e^{-\left(\log\frac{1}{ |E_1-E_2|}\right)^{1-\varepsilon}},
\end{equation*}
provided that $|E_1-E_2|$ is sufficiently small.
\end{theorem}

Using Theorem \ref{thmunew} instead of  Theorem \ref{thmu}, 
the proof of Theorems \ref{ldtfeb2new}, \ref{thmapp2new} and \ref{thmapp2'new} follows  from that of
Theorems \ref{thmapp4}, \ref{thmapp2} and \ref{thmapp2'} respectively. In order to avoid repetitions, we skip the details.

\section{Multi-scale analysis}
\subsection{Exhaustion construction for  an elementary region}

For $m\in \Z^d$ and  $ \Lambda\subset\mathbb{Z}^d$, define the distance by
$$\mathrm{dist}(m,\Lambda)=\inf_{n\in \Lambda}|m-n| .$$

Fix an elementary region $\Lambda\in \mathcal{E}_N$. Let $x\in \Lambda$.
Given $M\leq N/10$, we will construct exhaustion  at $x$ with width $M$.
Set
\begin{eqnarray*}
  \tilde{S}_0(x) &=& (x+[-2M,2M]^d)\cap \Lambda\\
  \tilde{S}_j(x) &=& \bigcup_{y\in S_{j-1}(x)}(y+[-4M,4M]^d)\cap \Lambda, 1\leq j\leq \tilde{l}
\end{eqnarray*}
where $\tilde{l}$ is the  minimum  such that $\tilde{S}_{\tilde{l}}(x)=\Lambda$.
We  set $S_{-1}(x)=\emptyset$ for convenience.

When $\tilde{S}_{j-1}(x)$ is very close to the boundary of $ \Lambda$,
$\tilde{A}_{j}(x)=\tilde{S}_{j}(x)\backslash \tilde{S}_{j-1}(x)$ and $\tilde{S}_j$ may have width less than $M $. However, there are at most finitely many $j$ with $0\leq j \leq \tilde{l}$, saying $C(d)$, such that $\tilde{A}_{j}(x)=\tilde{S}_{j}(x)\backslash \tilde{S}_{j-1}(x)$ has width less than $M$, where $C(d)$ is a constant depending on $d$.

We will delete $j$ if $\tilde{A}_{j}(x)=\tilde{S}_{j}(x)\backslash \tilde{S}_{j-1}(x)$ has small width and then rearrange exhaustion.  Here are the details.
Let $j_0\in \{0,1,\cdots, \tilde{l}-1\}$  be the possibly smallest number such that  both  $\tilde{S}_{j_0}(x)$ and $ \tilde{S}_{\tilde{l}}(x)\backslash\tilde{S}_{j_0}(x)$ have width at least $M$.  Otherwise, set  $j_0=\tilde{l}$. Let $S_0(x)=\tilde{S}_{j_0}(x)$.
Let $j_1 \in \{j_0,j_0+1,\cdots, \tilde{l} -1\}$  be the  possibly smallest number  such that  both $\tilde{S}_{j_1}(x)\backslash\tilde{S}_{j_0}(x)$  and  $\tilde{S}_{\tilde{l}}(x)\backslash\tilde{S}_{j_1}(x)$  have width at least $M$.  Otherwise, set  $j_1=\tilde{l}$.
Let $S_1(x)=\tilde{S}_{j_1}(x)$. Suppose we have defined $j_0$, $j_1\cdots, j_k$ and  corresponding $S_1(x), S_2(x),\cdots S_k(x)$.
Let $j_{k+1} \in \{j_k,j_k+1,\cdots, \tilde{l}-1\}$  be the  possibly smallest number  such that  $\tilde{S}_{j_{k+1}}(x)\backslash\tilde{S}_{j_k}(x)$ and $\tilde{S}_{\tilde{l}}(x)\backslash\tilde{S}_{j_{k+1}}(x)$
  have width at least $M$.  Otherwise, set $j_{k+1}=\tilde{l}$. Let $S_{k+1}(x)=\tilde{S}_{j_{k+1}}(x)$.
Let $l$ be such that $S_l(x)=\Lambda$.
 By our constructions,  
  $\tilde{l}-C(d)\leq l\leq \tilde{l}$.

  Here is an example. Assume  $x$ locates exactly  at the left upmost corner. In Fig.3, $\tilde{A}_k(x)=\tilde{S}_{k}(x)\backslash \tilde{S}_{k-1}(x)$ and  $\tilde{S}_{\tilde{l}}(x)=\tilde{S}_{\tilde{l}}(x)\backslash \tilde{S}_{\tilde{l}-1}(x)$  are the only two annuli which have width less than $M$.
  Therefore,
  \begin{itemize}
    \item  $ l=\tilde{l}-2$
    \item For $j=0, 1,2,\cdots,k-2$, $ S_j(x)=\tilde{S}_{j}(x)$.
    \item  For $j=k-1,k-2,\cdots, l-3$, $ S_j(x)=\tilde{S}_{j+1}(x)$. $ S_{l-2}(x)=\tilde{S}_{\tilde{l}}(x)$.
  \end{itemize}
For any elementary region $\Lambda$, $x\in \Lambda$ and $M$, we call $\{S_j(x)\}_{j=0}^{l}$ the exhaustion of $\Lambda$ at  $x$ with width $M$.
We call $A_j(x)=S_j(x)\backslash S_{j-1}(x)$ the $j$th annulus. For any $y\in S_j(x)\backslash S_{j-1}(x)$, $j=1,2,\cdots,l$, one has
\begin{equation}\label{Gdist}
4(j-1)M\leq |y-x|\leq 4jM+C(d)M.
\end{equation}

	By our constructions,  any $\{A_j(x)\}$  has width at least $M$. Namely,  for any  $n\in A_j(x)$ there exists  $W(n)\in \mathcal{E}_M$ such that
\begin{equation*}
n\in  W(n) \subset A_j(x)
\end{equation*}
and
\begin{equation*}
\text{ dist }(n,  A_j(x)\backslash  W(n))\geq M/2.
\end{equation*}
\begin{center}

\begin{tikzpicture}[thick, scale=2.4]
\draw (0,0)--(1.98,0)--(1.98,-1.98)--(3.96,-1.98)--(3.96,-3.96)--(0,-3.96)--(0,0);
\draw[](0,0)rectangle (0.1,-0.1);
\draw[](0,0)rectangle (0.3,-0.3);
\draw[](0,0)rectangle (0.5,-0.5);
\draw[](0,0)rectangle (0.7,-0.7);
\draw[](0,0)rectangle (0.9,-0.9);
\draw[](0,0)rectangle (1.1,-1.1);
\draw[](0,0)rectangle (1.3,-1.3);
\draw[](0,0)rectangle (1.5,-1.5);
\draw[](0,0)rectangle (1.7,-1.7);
\draw[dashed] (0,0)rectangle (1.9,-1.9);

\draw(0,-2.1)--(2.1,-2.1)--(2.1,-1.98);
\draw (0,-2.3)--(2.3,-2.3)--(2.3,-1.98);
\draw (0,-2.5)--(2.5,-2.5)--(2.5,-1.98);
\draw (0,-2.7)--(2.7,-2.7)--(2.7,-1.98);
\draw (0,-2.9)--(2.9,-2.9)--(2.9,-1.98);
\draw (0,-3.1)--(3.1,-3.1)--(3.1,-1.98);
\draw (0,-3.3)--(3.3,-3.3)--(3.3,-1.98);
\draw (0,-3.5)--(3.5,-3.5)--(3.5,-1.98);
\draw  (0,-3.7)--(3.7,-3.7)--(3.7,-1.98);
\draw [dashed](0,-3.9)--(3.9,-3.9)--(3.9,-1.98);


\draw (2.05,-4.1) node[below]   {Fig.3: Exhaustion Construction };
\draw (0,0) node[above]   { $x$};
\draw [->](0.3,0.3)--(0.05,-0.05);
\draw (0.3,0.4) node    { $\tilde{S}_0(x)$};

\draw (0.9,-1.8) node    { $\tilde{A}_{k-1}(x)$};
\draw (1,-2) node    { $\tilde{A}_k(x)$};
\draw (1.1,-2.2) node    { $\tilde{A}_{k+1}(x)$};
\draw (2.3,-3.8) node    { $\tilde{A}_{\tilde{l}-1}(x)$};
\draw (2.1,-3.6) node    { $\tilde{A}_{\tilde{l}-2}(x)$};
\draw [->](4.2,-3.22)--(3.92,-3.2);
\draw (4.2,-3.22) node[right]    { $\tilde{A}_{\tilde{l}}(x)$};




\end{tikzpicture}
\end{center}

\subsection{Resolvent identities}
For simplicity, assume $K=1$, namely
\begin{equation}\label{GOk1}
|A(n,n')|\leq   e^{-c_1|n-n'|^{\tilde{\sigma}}},  0<\tilde{\sigma}\leq 1, c_1>0,
\end{equation}
for any $n,n^\prime\in \Z^d$.
For any $\Lambda\subset \Z^d$, denote by $A_{\Lambda}=R_{\Lambda} AR_{\Lambda}$, where $R_{\Lambda}$ is the restriction on $\Lambda$, and the Green's function
\begin{equation*}
  G_{\Lambda}=(R_{\Lambda} AR_{\Lambda})^{-1},
\end{equation*}
provided $R_{\Lambda} AR_{\Lambda}$ is invertible.  Denote  by $G_{\Lambda}(n,n^\prime)$ its  elements, $n,n^\prime\in\Lambda\subset \Z^d$.

Assume $\Lambda_1$ and $\Lambda_2$ are wo disjoint subsets of  $\Z^d$. 
Namely,  $\Lambda_1,\Lambda_2 \subset \Z^d$ and $\Lambda_1\cap\Lambda_2=\emptyset$. Let $\Lambda=\Lambda_1\cup \Lambda_2$.
 Suppose that $ R_{\Lambda}AR_{\Lambda}$ and $ R_{\Lambda_i}AR_{\Lambda_i}$, $i=1,2$ are invertible.
 Then
 \begin{equation*}
   G_{\Lambda}=G_{\Lambda_1}+G_{\Lambda_2}-(G_{\Lambda_1}+G_{\Lambda_2})( A_{\Lambda}-A_{\Lambda_1}-A_{\Lambda_2})G_{\Lambda}.
 \end{equation*}
 If $m\in \Lambda_1$ and $n\in \Lambda$, we have
 \begin{equation}\label{Greso}
    |G_{\Lambda}(m,n)|\leq |G_{\Lambda_1}(m,n)|\chi_{\Lambda_1}(n)+ \sum_{n^{\prime}\in \Lambda_1,n^{\prime\prime}\in \Lambda_2} e^{-c_1|n^{\prime}-n^{\prime\prime}|^{\tilde{\sigma}}}|G_{\Lambda_1}(m,n^{\prime})||G_{\Lambda}(n^{\prime\prime},n)|.
 \end{equation}
  If $n\in \Lambda_2$ and $m\in \Lambda$, we have
 \begin{equation}\label{Greson}
    |G_{\Lambda}(m,n)|\leq |G_{\Lambda_2}(m,n)|\chi_{\Lambda_2}(n)+ \sum_{n^{\prime}\in \Lambda_1,n^{\prime\prime}\in \Lambda_2} e^{-c_1|n^{\prime}-n^{\prime\prime}|^{\tilde{\sigma}}}|G_{\Lambda}(m,n^{\prime})||G_{\Lambda_2}(n^{\prime\prime},n)|.
 \end{equation}

\begin{lemma}[Schur  test]\label{schur}
Suppose $A=A_{ij}$ is a    matrix. Then
\begin{equation*}
  \|A\|\leq \sqrt{\left(\sup_{i}\sum_{j}|A_{ij}|\right)\left(\sup_{j}\sum_{i}|A_{ij}|\right)}.
\end{equation*}
\end{lemma}
The following lemma is a generalization of  Lemma 3.2 in \cite{jls}.
\begin{lemma}\label{res1}
 Let $ {c_2}\in [\tilde{c}_1,c_1]$, $\sigma<\tilde{\sigma}$ and $M_0\leq M_1\leq N$.  Assume $\Lambda$ is a subset of $\Z^d$ with  ${\rm diam}(\Lambda)\leq 2N+1$.  Suppose that for any $n\in \Lambda $, there exists some  $ W=W(n)\in \mathcal{E}_M$ with
$M_0\leq M\leq M_1$ such that
$n\in W\subset \Lambda$,  ${\rm dist} (n,\Lambda \backslash W)\geq \frac{M}{2}$ and
\begin{eqnarray}
\label{w1}&& \|G_{W(n)}\|\leq2 e^{M^\sigma},\\
\label{w2}&& |G_{W(n)}(n,n')|\leq  2e^{- {c}_2|n-n'|^{\tilde{\sigma}}}\  {\mathrm{for} \ |n-n'|\geq \frac{M}{10}}.
\end{eqnarray}
  We assume further that $M_0$ is large enough so that
\begin{equation}\label{ml}
\sup_{M_0\leq M\leq M_1} \sup_{ {c_2}\in [\tilde{c}_1,c_1]}2e^{{M}^\sigma}(2M+1)^{d}e^{\frac{{c}_2}{10^{\tilde{\sigma}}}M^{\tilde{\sigma}}}\sum_{j=0}^{\infty}(M+2j+1)^de^{-{c}_2(j+M/2)^{\tilde{\sigma}}}\leq \frac{1}{2}.
\end{equation}
Then
\begin{equation*}
  \|G_{\Lambda}\|\leq 4 (2M_1+1)^d e^{ M_1^\sigma}.
\end{equation*}
\end{lemma}
\begin{proof}

Under the assumption of \eqref{ml}, it is easy to check that for any $M$ with  $M_0\leq M\leq M_1$ and any $n\in \Lambda$,
\begin{equation}\label{Large1}
 2 (2M+1)^de^{{M}^\sigma+\frac{c_2}{10^{\tilde{\sigma}}}M^{\tilde{\sigma}}}\sum_{n_2\in \Lambda\atop |n_2-n|\geq \frac{M}{2}}   e^{-c_2|n-n_2|^{\tilde{\sigma}}}\leq \frac{1}{2}.
\end{equation}

By \eqref{w1} and \eqref{w2},
one has
\begin{equation}\label{Ugood}
   |G_{W(n)}(n,n')|\leq 2e^{M^\sigma+\frac{{c_2}}{10^{\tilde{\sigma}}}M^{\tilde{\sigma}}}e^{- {c}_2|n-n'|^{\tilde{\sigma}}}.
\end{equation}
For each $n\in \Lambda$, applying \eqref{Greso} with $\Lambda_1=W(n)$, one has
\begin{equation*}
  |G_{\Lambda}(n,n')|\leq |G_{W(n)}(n,n')|\chi_{W(n)}(n')+ \sum_{n_1\in W(n)\atop n_2\in \Lambda\backslash W(n)} e^{-c_1|n_1-n_2|^{\tilde{\sigma}}}|G_{W(n)}(n,n_1)||G_{\Lambda}(n_2,n')|.
\end{equation*}
It is easy to see for $0<\tilde{\sigma}\leq 1$,
\begin{equation}\label{gdec101}
|x+y|^{\tilde{\sigma}}\leq |x|^{\tilde{\sigma}}+|y|^{\tilde{\sigma}}.
\end{equation}
By \eqref{Ugood} and the fact that $|W(n)|\leq (2M+1)^d$, one has
\begin{align}
  |G_{\Lambda}(n,n')| \leq &|G_{W(n)}(n,n')|\chi_{W(n)}(n')\nonumber\\
  &+
  2 \sum_{n_1\in W(n)\atop n_2\in \Lambda\backslash W(n)} e^{{M}^\sigma+\frac{c_2}{10^{\tilde{\sigma}}}M^{\tilde{\sigma}}}e^{-c_2|n-n_1|^{\tilde{\sigma}}} e^{-c_1|n_1-n_2|^{\tilde{\sigma}}}|G_{\Lambda}(n_2,n')|\nonumber   \\
  \leq &|G_{W(n)}(n,n')|\chi_{W(n)}(n') \nonumber\\
  &+2(2M+1)^de^{{M}^\sigma
  	+\frac{c_2}{10^{\tilde{\sigma}}}M^{\tilde{\sigma}}}\sum_{n_2\in \Lambda\backslash W(n)}   e^{-c_2|n-n_2|^{\tilde{\sigma}}}|G_{\Lambda}(n_2,n')| \nonumber\\
  \leq&|G_{W(n)}(n,n')|\chi_{W(n)}(n')\nonumber\\
  &+2(2M+1)^de^{{M}^\sigma+\frac{c_2}{10^{\tilde{\sigma}}}M^{\tilde{\sigma}}}\sum_{n_2\in \Lambda\atop |n_2-n|\geq \frac{M}{2}}   e^{-c_2|n-n_2|^{\tilde{\sigma}}}|G_{\Lambda}(n_2,n')|
  \label{BGSle1}.
\end{align}
where the second inequality holds by \eqref{gdec101} and the  last inequality holds by the assumption ${\rm dist} (n,\Lambda\backslash W(n))\geq \frac{M}{2}$.

Summing over $n'\in \Lambda$ in \eqref{BGSle1} and noticing \eqref{Large1} yields
\begin{eqnarray}\label{sch1}
 \sup_{n\in \Lambda}\sum_{n'\in \Lambda} |G_{\Lambda}(n,n')|
  &\leq&2(2M_1+1)^de^{{M}_1^\sigma}+\frac{1}{2}\sup_{n_2\in \Lambda}\sum_{n'\in \Lambda}|G_{\Lambda}(n_2,n')|.
\end{eqnarray}
Similarly, using \eqref{Greson} instead of \eqref{Greso}, one has
\begin{equation}\label{sch2}
 \sup_{n\in \Lambda} \sum_{n'\in \Lambda} |G_{\Lambda}(n',n)|\leq  2(2M_1+1)^de^{{M}_1^\sigma}+\frac{1}{2}\sup_{n_2\in \Lambda}\sum_{n'\in \Lambda}|G_{\Lambda}(n',n_2)|.
\end{equation}

Now the lemma follows from \eqref{sch1}, \eqref{sch2} and  Lemma \ref{schur}.

\end{proof}

\subsection{ Proof of Theorem \ref{thmmul}}
\begin{proof}
Choose a constant $\rho\in(1,1+ \tilde{\sigma}-\sigma)$. Calculation shows $\rho\sigma<\tilde{\sigma}$.

Define inductively $M_{j+1}=\lfloor M_j^\rho\rfloor$, $M_0=M$.
Let $\gamma_0=c_2$.
Fix an elementary region $\tilde{\Lambda}_1\in \mathcal{E}_{M_1}$ and $\tilde{\Lambda}_1\subset \tilde{\Lambda}_0$. For any $x\in \tilde{\Lambda}_1$, consider the exhaustion $\{S_j(x)\}_{j=0}^l$ of $\tilde{\Lambda}_1$  at $x$  with width $M_0$. Denote by $\{A_k(x)\}$  the annuli.

 We call the annulus  $A_k(x)$  good, if for any $y\in A_k(x)$, there exists $W(y)\in \mathcal{E}_{{M}_0}$ such that
\begin{equation*}
  y\in W(y)\subset A_k(x), \text { dist}(y, A_k(x)\backslash W(y))\geq M_0/2,
\end{equation*}
and  for $|n-n^\prime|\geq \frac{M_0}{10}$,
\begin{equation}\label{ggoodnov26}
  |(R_{W(y)}A  R_{W(y)})^{-1}(n,n^\prime)|\leq  e^{-\gamma_0|n-n^\prime|^{\tilde{\sigma}}}.
\end{equation}
Otherwise, we call the annulus $A_k$ bad.

Fix   $\kappa>0$, which  will be determined later.
An elementary region $\tilde{\Lambda}_1\subset \tilde{\Lambda}_0$ is called  bad if  there exists  $x\in \tilde{\Lambda}_1$  such that the number of bad annuli $\{A_k(x)\}$
exceeds
\begin{equation*}
  B_1:=\kappa \frac{M_1}{M_0}.
\end{equation*}
Otherwise, we call   $\tilde{\Lambda}_1$ good.
Let $\mathcal{F}_1$ be an arbitrary family of pairwise disjoint bad  elementary regions in $\mathcal{E}_{M_1}$ contained in $\tilde{\Lambda}_0$.
Since every annulus in $\{A_k\}$ has width at least $M_0$ by our construction, one has  that  every bad annulus contains at least one elementary region in $\mathcal{E}_{M_0}$ without satisfying \eqref{ggoodnov26} and  hence
\begin{equation}\label{gnumber1}
 \# \mathcal{F}_1\leq \frac{N^{\varsigma}}{\kappa M_1}.
\end{equation}
Assume that $\tilde{\Lambda}_1\subset \tilde{\Lambda}_0$ is a good elementary region in $\mathcal{E}_{M_1}$.
 We will first show that $\tilde{\Lambda}_1$ is in class G with slightly smaller $\gamma_0$.
Consider the exhaustion $\{S_j(x)\}$ of $\tilde{\Lambda}_1$ at $x$ with width $M_0$. By the assumption, there are no more than $B_1$ bad annuli in this exhaustion.
Denote by $\{A_j(x)\}_{j=0}^{{l}}$ annuli.
By putting adjacent  good annuli or bad annuli together,
we obtain   a new exhaustion
\begin{equation}\label{ggdec12124}
  \emptyset={J}_{-1}\subset {J}_0\subset {J}_1\subset \cdots \subset{J}_{{g}}=\tilde{\Lambda}_1.
\end{equation}

More precisely, $\{J_s(x)\}$, $s=0,1,2,\cdots,g$, satisfies   the following rules.
\begin{itemize}
  \item    $J_s(x)\backslash J_{s-1}(x)=\{A_j(x)\}_{j=t_s}^{j=t_s^\prime}$  for some $t_s<t_s^\prime$. 
  \item $x\in J_0(x)$.
\item  The annuli $A_j(x)$, $j=t_s,t_s+1,\cdots,t_s^\prime$  are either all good or all bad.
 \item Take $J_s(x)$ maximal with the above three properties.
\end{itemize}
We remind that $J_s(x)\backslash J_{s-1}(x)$ has width at least $M_0$ for any $s=0,1,2,\cdots,g$.
By our construction, if all annuli in  $J_s(x)\backslash J_{s-1}(x)$  are good (bad), then all annuli in $J_{s+1}(x)\backslash J_{s}(x)$  are bad (good).

For any $n\in \tilde{\Lambda}_1$, let $k(n)$   be the number of good  annuli  between $x$ and $n$. Namely, for any 
 $n\in A_j(x)$,  
 \begin{equation*}
 k(n)=\#\{A_t(x): A_t(x) \text{ is a good annulus}, 0\leq t\leq j\}.
 \end{equation*}
 Before we start the estimates, let us give  several facts  first, which will be used  constantly  in the  later proof.
 By our constructions, $J_s$ is a generalized elementary region,  $s=0,1,\cdots, g$.
 By the assumption \eqref{gboundnov26}, one has for all $s=0,1,\cdots, g$,
 \begin{equation}\label{gboundnov261}
 ( R_{J_{s}} A R_{J_{s}} )^{-1}\leq e^{M_1^\sigma}.
 \end{equation}
 Assume
 \begin{equation*}
 0< c\leq (1-5^{-\tilde{\sigma}})c_1.
 \end{equation*}
 If $|n-n_2|\geq\frac{M}{2}$ and $|n-n_1|\leq \frac{M}{10}$, one has
 \begin{eqnarray}
 c_1|n_1-n_2|^{\tilde{\sigma}} &\geq& c_1(|n-n_2|^{\tilde{\sigma}}-|n-n_1|^{\tilde{\sigma}})\nonumber \\
 &\geq& c|n-n_2|^{\tilde{\sigma}} .\label{gdec94}
 \end{eqnarray}
 It is clear that  for any  $n_1,n_2\in \tilde{\Lambda}_1$, 
 \begin{equation}\label{gggdec1133}
 4M_0 k(n_1)+|n_1-n_2|\geq  4M_0k(n_2)-4M_0.
 \end{equation}
 
 Without loss of generality, assume all the annuli in $J_0$   are bad (the another case is  similar).

For any $n\in \tilde{\Lambda}_1$, define
\begin{equation*}
  \Gamma_s(n)=\max\{4M_0k(n)-10(s+1)M_0,0\}.
\end{equation*}

By   \eqref{gggdec1133}, we have for any  $n_1\in \tilde{\Lambda}_1$ and $n_2\in \tilde{\Lambda}_1$,
\begin{equation}\label{gggdec131}
    \Gamma_s(n_1)+|n_2-n_1|\geq \max\{  \Gamma_s(n_2)-4M_0,0\}.
\end{equation}

 We shall   inductively obtain estimates of the form
 \begin{equation}\label{gind}
   |G_{J_s(x)}(x,z)|\leq T_se^{-\gamma_0\Gamma_s^{\tilde{\sigma}}(z)},
 \end{equation}
 where $z\in J_s$, $s=0,1,\cdots,g$.


{\bf First step: $s=0$}

 Since all annuli in $J_0$ are bad, one has $k(z)=0$ and hence
 \begin{equation}\label{Gnov261}
   \Gamma_0(z)=0.
 \end{equation}

By \eqref{gboundnov261} and \eqref{Gnov261}, one has   for $z\in J_0(x)$,
\begin{eqnarray*}
   |G_{J_s}(x,z)| &\leq& e^{M_1^\sigma} \\
   &= & e^{M_1^\sigma} e^{-\gamma_0 \Gamma_0^{\tilde{\sigma}}(z)}.
\end{eqnarray*}
It implies that \eqref{gind} holds for
\begin{equation}\label{gggdec111}
  T_0=e^{M_1^{\sigma}}.
\end{equation}
Assume \eqref{gind} holds at $s$th step for a proper $T_s$.

{\bf Case 1:}
All  annuli in  $J_{s+1}\backslash J_{s}$  are bad. 

Pick any $z\in J_{s+1}$.
Let $\tilde{n}_1\in J_s $ and $ \tilde{n}_2\in J_{s+1}\backslash J_s$ be such that
\begin{equation*}
\Gamma_s(\tilde{n}_1)+|\tilde{n}_1-\tilde{n}_2|= \inf_{ n_1\in J_s \atop n_2\in J_{s+1}\backslash J_s} (\Gamma_s(n_1)+|n_1-n_2|).
\end{equation*}

{\bf Case $1_1$}: $z\in J_{s+1}\backslash J_s$.
In this case, for any $n_2\in J_{s+1}\backslash J_s$, one has
\begin{equation}\label{gggdec114}
  k(z)=k(n_2), \Gamma_s(z)=\Gamma_s(n_2),
\end{equation}
since all  annuli in $J_{s+1}\backslash J_s$ are bad.

 Applying \eqref{Greso}  ($\Lambda_1=J_s$ and $\Lambda_2=J_{s+1}\backslash J_s$), one has
\begin{eqnarray}
  |G_{J_{s+1}}(x,z)| &\leq&   \sum_{n_1\in J_s \atop n_2\in J_{s+1}\backslash J_s} |G_{J_s}(x,n_1)| e^{-c_1|n_1-n_2|^{\tilde{\sigma}}}|G_{J_{s+1}}(n_2,z)|\nonumber  \\
   &\leq&  \sum_{n_1\in J_s \atop n_2\in J_{s+1}\backslash J_s}  T_se^{- \gamma_0\Gamma_s^{\tilde{\sigma}}(n_1)}  e^{-\gamma_0|n_1-n_2|^{\tilde{\sigma}}}|G_{J_{s+1}}(n_2,z)|\nonumber  \\
  &\leq& e^{M_1^{\sigma}}T_s\sum_{n_1\in J_s \atop n_2\in J_{s+1}\backslash J_s} e^{- \gamma_0\Gamma_s^{\tilde{\sigma}}(n_1)}  e^{-\gamma_0|n_1-n_2|^{\tilde{\sigma}}}  \nonumber\\
   &\leq&  (2M_1+1)^{2d}e^{M_1^{\sigma}}T_s \sup_{n_1\in J_s \atop n_2\in J_{s+1}\backslash J_s}  e^{-\gamma_0\Gamma_s^{\tilde{\sigma}}(n_1)}  e^{-\gamma_0|n_1-n_2|^{\tilde{\sigma}}} \nonumber\\
   &\leq& (2M_1+1)^{2d}    e^{M_1^\sigma}  T_s  e^{-\gamma_0 (\Gamma_s(\tilde{n}_1)+|\tilde{n}_1-\tilde{n}_2|)^{\tilde{\sigma}}}     \nonumber\\
    &\leq&(2M_1+1)^{2d}   e^{M_1^\sigma}  T_s   e^{-\gamma_0 (\max\{\Gamma_{s}(\tilde{n}_2)-4M_0,0\})^{\tilde{\sigma}}}, \nonumber\\
    &\leq&(2M_1+1)^{2d}   e^{M_1^\sigma}  T_s   e^{-\gamma_0 (\max\{\Gamma_{s}(z)-4M_0,0\})^{\tilde{\sigma}}}, \label{e12132}
\end{eqnarray}
where the second inequality holds by the induction \eqref{gind} and $\gamma_0\leq c_1$, the third inequality holds by \eqref{gboundnov261}, the fifth inequality holds by \eqref{gdec101},
the sixth inequality holds
 \eqref{gggdec131}, and the last inequality holds by \eqref{gggdec114}.


{\bf Case $1_2$}: $z\in J_s$. 

 In this case, we have  for any $n_2\in J_{s+1}\backslash J_s$,
\begin{equation}\label{gggdec114new}
 k(n_2)\geq   k(z).
\end{equation}

 Applying \eqref{Greso}  ($\Lambda_1=J_s$ and $\Lambda_2=J_{s+1}\backslash J_s$), one has
\begin{eqnarray}
  |G_{J_{s+1}}(x,z)| &\leq& |G_{J_s}(x,z)| +  \sum_{n_1\in J_s \atop n_2\in J_{s+1}\backslash J_s} |G_{J_s}(x,n_1)| e^{-c_1|n_1-n_2|^{\tilde{\sigma}}}|G_{J_{s+1}}(n_2,z)|\nonumber  \\
   &\leq&  T_se^{-\gamma_0\Gamma_s^{\tilde{\sigma}}(z)}+\sum_{n_1\in J_s \atop n_2\in J_{s+1}\backslash J_s}  T_se^{-\gamma_0\Gamma_s^{\tilde{\sigma}}(n_1)}  e^{-\gamma_0|n_1-n_2|^{\tilde{\sigma}}}|G_{J_{s+1}}(n_2,z)|\nonumber  \\
  &\leq&  T_se^{-\gamma_0\Gamma_s^{\tilde{\sigma}}(z)}+ (2M_1+1)^{2d}   e^{M_1^\sigma}   T_s  \sup_{n_1\in J_s\atop n_2\in J_{s+1}\backslash J_s}  e^{- \gamma_0(\Gamma_{s}(n_1)+|n_1-n_2|)^{\tilde{\sigma}}}     \nonumber\\
    &=&  T_se^{-\gamma_0\Gamma_s^{\tilde{\sigma}}(z)}+ (2M_1+1)^{2d}   e^{M_1^\sigma}   T_s    e^{- \gamma_0(\Gamma_{s}(\tilde{n}_1)+|\tilde{n}_1-\tilde{n}_2|)^{\tilde{\sigma}}}     \nonumber\\
    &\leq&T_se^{-\gamma_0\Gamma_s^{\tilde{\sigma}}(z)}+(2M_1+1)^{2d}   e^{M_1^\sigma}  T_s   e^{-\gamma_0 (\max\{\Gamma_{s}(\tilde{n}_2)-4M_0,0\})^{\tilde{\sigma}}}, \nonumber\\
     &\leq&T_se^{-\gamma_0\Gamma_s^{\tilde{\sigma}}(z)}+(2M_1+1)^{2d}   e^{M_1^\sigma}  T_s   e^{-\gamma_0 (\max\{\Gamma_{s}(z)-4M_0,0\})^{\tilde{\sigma}}}\nonumber\\
       &\leq&2(2M_1+1)^{2d}   e^{M_1^\sigma}  T_s   e^{-\gamma_0 (\max\{\Gamma_{s}(z)-4M_0,0\})^{\tilde{\sigma}}},\label{e12134}
\end{eqnarray}
where the second inequality holds by the induction \eqref{gind}, the third inequality holds by \eqref{gboundnov261} and \eqref{gdec101},  the forth inequality holds by
 \eqref{gggdec131}, and the fifth  inequality holds by \eqref{gggdec114new}.

{\bf Case 2:}  All the annuli in $J_{s+1}\backslash J_s$ are good.

By our constructions, 
 $J_{s+1}\backslash J_s$ has width at least $M_0$. 
Therefore,
for any $k\in J_{s+1}\backslash J_s$,
there exists some  $ W=W(k)\in \mathcal{E}_{M_0}$ such that
$k\in W\subset \Lambda$,  
\begin{equation}\label{gm2}
{\rm dist} (k,J_{s+1}\backslash J_s \backslash W)\geq \frac{M_0}{2},
\end{equation}
and
\begin{eqnarray}
\label{w1nov}&& \|G_{W(k)} \|\leq e^{M_0^\sigma},\\
\label{w2nov}&& |G_{W(k)} (n_1,n_2)|\leq  e^{-\gamma_0|n_1-n_2|^{\tilde{\sigma}}}\  {\mathrm{for} \ |n_1-n_2|\geq \frac{M_0}{10}},
\end{eqnarray}
where \eqref{w1nov} holds by the assumption \eqref{gboundnov26}.

Since $M_0$ is large enough, one has \eqref{ml} is satisfied.  Applying
 Lemma \ref{res1}, we have
\begin{equation}\label{gnov265}
  \|G_{J_{s+1}\backslash J_s}\|\leq 4 (2M_0+1)^d e^{ M_0^\sigma}.
\end{equation}
We remark that we can not use the assumption \eqref{gboundnov26} to bound $G_{J_{s+1}\backslash J_s}$ since $ J_{s+1}\backslash J_s$ is not necessary to be a generalized elementary region.  It is worth to point out that $ J_{s+1}\backslash J_s$ may not be connected.

We will first prove that for any $m,n\in J_{s+1}\backslash J_s$,
\begin{equation}\label{gdefm}
|G_{J_{s+1}\backslash J_s}(m,n)| \leq  M_1^{10d  M_1^{\tilde{\sigma}} M_0^{\sigma-\tilde{\sigma}}}   e^{-  \gamma_0(\max\{|m-n| -2M_0,0\})^{\tilde{\sigma}}} .
\end{equation}
Assume $|m-n|\leq 2M_0$.  \eqref{gdefm} holds by \eqref{gnov265}.

Assume $|m-n|>2M_0$.
Applying  \eqref{Greso} with $\Lambda_1=W(m)$ and   using that $|m-n|> 2M_0$, one has
\begin{equation}\label{BGSlenov}
  |G_{J_{s+1}\backslash J_s}(m,n)|\leq  \sum_{n_1\in W(m)\atop n_2\in J_{s+1}\backslash J_s\backslash W(m)} e^{-c_1|n_1-n_2|^{\tilde{\sigma}}}|G_{W(m)}(m,n_1)||G_{\Lambda}(n_2,n)|.
\end{equation}
Applying \eqref{gm2} with $k=m$ and by \eqref{gdec94}, one has
for any $n_1$  with $|n_1-m|\leq \frac{M_0}{10}$ and $n_2\in J_{s+1}\backslash J_s\backslash W(m)$, one has 
\begin{equation}\label{gm3}
c_1|n_1-n_2|^{\tilde{\sigma}}\geq c_2|m-n_2|^{\tilde{\sigma}}.
\end{equation}
By \eqref{w1nov}, \eqref{w2nov} and \eqref{BGSlenov}, we have
\begin{align}
\nonumber|G_{J_{s+1}\backslash J_s}&(m,n)|\\
\leq &\sum_{n_1\in W(m),|n_1-m|\leq \frac{M_0}{10}-1\atop n_2\in J_{s+1}\backslash J_s\backslash W(m)} e^{-c_1|n_1-n_2|^{\tilde{\sigma}}}|G_{W(m)}(m,n_1)||G_{J_{s+1}\backslash J_s}(n_2,n)| \nonumber\\
\nonumber&+ \sum_{n_1\in W(m),|n_1-m|\geq \frac{M_0}{10}\atop n_2\in J_{s+1}\backslash J_s\backslash W(m)} e^{-c_1|n_1-n_2|^{\tilde{\sigma}}}|G_{W(m)}(m,n_1)||G_{J_{s+1}\backslash J_s}(n_2,n)|\\
\leq &\nonumber  \sum_{n_1\in W(m),|n_1-m|\leq \frac{M_0}{10}-1\atop n_2\in J_{s+1}\backslash J_s\backslash W(m)} e^{{M}_0^\sigma}e^{-c_1|n_1-n_2|^{\tilde{\sigma}}}|G_{J_{s+1}\backslash J_s}(n_2,n)|\\
\nonumber&+ \sum_{n_1\in W(m),|n_1-m|\geq \frac{M_0}{10}\atop n_2\in J_{s+1}\backslash J_s\backslash W(m)} e^{-c_1|n_1-n_2|^{\tilde{\sigma}}} e^{-\gamma_0 |m-n_1|^{\tilde{\sigma}}}|G_{J_{s+1}\backslash J_s}(n_2,n)|\\
\leq& \nonumber \sum_{n_1\in W(m),|n_1-m|\leq \frac{M_0}{10}-1\atop n_2\in J_{s+1}\backslash J_s\backslash W(n_2)} e^{M_0^\sigma} e^{-\gamma_0|m-n_2|^{\tilde{\sigma}}}|G_{J_{s+1}\backslash J_s}(n_2,n)|\\
\nonumber&+ \sum_{n_1\in W(m),|n_1-m|\geq \frac{M_0}{10}\atop n_2\in J_{s+1}\backslash J_s\backslash W(m)} e^{-\gamma_0|m-n_2|^{\tilde{\sigma}}} |G_{J_{s+1}\backslash J_s}(n_2,n)|\\
\label{rsi8nov}\leq&  (2M_1+1)^{2d}e^{M_0^\sigma} \sup_{n_2\in J_{s+1}\backslash J_s\backslash W(m)}  e^{-\gamma_0|m-n_2|^{\tilde{\sigma}} }|G_{J_{s+1}\backslash J_s}(n_2,n)|,
\end{align}
where the third inequality holds because of \eqref{gm3}. 

Recall that $ |m-n_2|\geq \frac{M_0}{2}$.
Iterating \eqref{rsi8nov} until $|n_2-n|\leq 2M_0$  or  at most  $\lfloor\frac{2^{\tilde{\sigma}}|m-n|^{\tilde{\sigma}}}{M_0^{\tilde{\sigma}}}\rfloor+1$ times,  we have
\begin{eqnarray}
  |G_{J_{s+1}\backslash J_s}(m,n)|  &\leq &  e^{M_0^\sigma \left(\frac{2^{\tilde{\sigma}}|m-n|^{\tilde{\sigma}}}{M_0^{\tilde{\sigma}}}+1\right)}(2M_1+1)^{ 2d\left(\frac{2^{\tilde{\sigma}}|m-n|^{\tilde{\sigma}}}{M_0^{\tilde{\sigma}}}+1\right) }   \nonumber \\
  &&\times
  e^{-   \gamma_0 (|m-n|-2M_0)^{\tilde{\sigma}} }  \|G_{J_{s+1}\backslash J_s}\|\nonumber \\
   &\leq &  M_1^{9d  M_1^{\tilde{\sigma}} M_0^{\sigma-\tilde{\sigma}}}   e^{-   \gamma_0 (|m-n|-2M_0)^{\tilde{\sigma}} }  \|G_{J_{s+1}\backslash J_s}\|\nonumber \\
  &\leq &  M_1^{9d  M_1^{\tilde{\sigma}} M_0^{\sigma-\tilde{\sigma}}}     e^{-   \gamma_0 (|m-n|-2M_0)^{\tilde{\sigma}} }  4 (2M_0+1)^d e^{ M_0^\sigma}\nonumber \\
  &\leq & M_1^{10d  M_1^{\tilde{\sigma}} M_0^{\sigma-\tilde{\sigma}}}  e^{-   \gamma_0 (|m-n|-2M_0)^{\tilde{\sigma}} }  ,\label{gnov2610}
\end{eqnarray}
where the first inequality holds by $|m-n|\leq 2 M_1$ and the third inequality holds by \eqref{gnov265}.

{\bf Case $2_1$:} $z\in J_s$. For this case, following the proof of  Case $1_2$ (see \eqref{e12134}), one has
\begin{eqnarray}
  |G_{J_{s+1}}(x,z)| &\leq&2(2M_1+1)^{2d}   e^{M_1^\sigma}  T_s   e^{-\gamma_0 (\max\{\Gamma_{s}(z)-4M_0,0\})^{\tilde{\sigma}}}. \label{e12135}
\end{eqnarray}

{\bf Case $2_2$:}  $z\in J_{s+1}\backslash J_s$.

Applying \eqref{Greson}  ($\Lambda_1=J_s$ and $\Lambda_2=J_{s+1}\backslash J_s$), one has
\begin{align}
 |G_{J_{s+1}}&(x,z)|  \nonumber \\
  \leq&  \sum_{n_1\in J_s \atop n_2\in J_{s+1}\backslash J_s} |G_{J_{s+1}\backslash J_s}(n_2,z)| e^{-c_1|n_1-n_2|^{\tilde{\sigma}}} |G_{J_{s+1}}(x,n_1)|\nonumber  \\
  \leq& 2(2M_1+1)^{2d}   e^{M_1^\sigma}  T_s   \nonumber\\
  &\times   \sum_{n_1\in J_s \atop n_2\in J_{s+1}\backslash J_s} |G_{J_{s+1}\backslash J_s}(n_2,z)| e^{-c_1|n_1-n_2|^{\tilde{\sigma}}} e^{-\gamma_0 (\max\{\Gamma_{s}(n_1)-4M_0,0\})^{\tilde{\sigma}}} \nonumber
  \\
  \leq&    2(2M_1+1)^{4d}   e^{M_1^{\sigma}}M_1^{10d  M_1^{\tilde{\sigma}} 
  	M_0^{\sigma-\tilde{\sigma}}}T_s \nonumber\\
   &\times \sup_{n_1\in J_s \atop n_2\in J_{s+1}\backslash J_s}  e^{-  \gamma_0(\max\{|n_2-z| -2M_0,0\})^{\tilde{\sigma}}}e^{-c_1|n_1-n_2|^{\tilde{\sigma}}} e^{-\gamma_0 (\max\{\Gamma_{s}(n_1)-4M_0,0\})^{\tilde{\sigma}}}\nonumber\\
 \leq&
   M_1^{11d  M_1^{\tilde{\sigma}} M_0^{\sigma-\tilde{\sigma}}}T_s\sup_{n_1\in J_s}  e^{-  \gamma_0(\max\{|n_1-z| -2M_0,0\})^{\tilde{\sigma}}} e^{-\gamma_0 (\max\{\Gamma_{s}(n_1)-4M_0,0\})^{\tilde{\sigma}}}\nonumber
  \\
   \leq& M_1^{11d  M_1^{\tilde{\sigma}} M_0^{\sigma-\tilde{\sigma}}} T_se^{-  \gamma_0(\max\{\Gamma_s(z) -10M_0,0\})^{\tilde{\sigma}}} , \label{e12137}
\end{align}
where the second inequality holds by \eqref{e12135}, the third inequality holds by \eqref{gdefm},  the  fourth inequality holds by \eqref{gdec101} and
 the fifth inequality holds by \eqref{gdec101} and \eqref{gggdec131}.

Putting all cases together and by  \eqref{e12132},   \eqref{e12134}, \eqref{e12135} and \eqref{e12137}, one has if \eqref{gind} holds at $s$th step, then
\begin{eqnarray}
   |G_{J_{s+1}}(x,z)| &\leq& M_1^{12d  M_1^{\tilde{\sigma}} M_0^{\sigma-\tilde{\sigma}}} T_se^{-  \gamma_0(\max\{\Gamma_s(z) -10M_0,0\})^{\tilde{\sigma}}}\nonumber \\
   &\leq& M_1^{12d  M_1^{\tilde{\sigma}} M_0^{\sigma-\tilde{\sigma}}} T_se^{-  \gamma_0\Gamma_{s+1}^{\tilde{\sigma}}(z) }. \label{gggdec12121}
\end{eqnarray}


By \eqref{gggdec111} and \eqref{gggdec12121},  we obtain that \eqref{gind} is true for
\begin{equation}\label{Gb0}
  T_0=e^{M_1^{\sigma}},
\end{equation}
and
\begin{equation}\label{Gbs}
T_{s+1}=   M_1^{12d  M_1^{\tilde{\sigma}} M_0^{\sigma-\tilde{\sigma}}}  T_s.
\end{equation}
By \eqref{gind}, \eqref{Gb0} and \eqref{Gbs}, one has
\begin{equation}\label{gindnewnew}
   |G_{J_g}(x,z)|\leq M_1^{13gd  M_1^{\tilde{\sigma}} M_0^{\sigma-\tilde{\sigma}}}e^{-  \gamma_0\Gamma_{g}^{\tilde{\sigma}}(z) }.
 \end{equation}
 By the assumption that $\tilde{\Lambda}_1$ is good, one has
  \begin{equation}\label{gnov2615}
 g\leq 2B_1=2\kappa\frac{M_1}{M_0} ,
 \end{equation}
 and hence (by \eqref{Gdist})
  \begin{equation}\label{gnov2615new}
k(z)\geq \frac{|x-z|}{4M_0}-2B_1-C(d)\geq \frac{|x-z|}{4M_0}-2\kappa \frac{M_1}{M_0}-C(d).
 \end{equation}
By \eqref{gnov2615new} and the definition of $\Gamma_s$, we have  for $|x-z|\geq \frac{M_1}{10}$,
\begin{eqnarray}
  \Gamma_{g}^{\tilde{\sigma}}(z) &\geq& (|x-z|-8\kappa M_1-10(g+1) M_0)^{\tilde{\sigma}}\nonumber  \\
   &\geq&   (|x-z|-30\kappa M_1)^{\tilde{\sigma}} \nonumber \\
    &\geq& |x-z| ^{\tilde{\sigma}} (1-160\kappa )^{\tilde{\sigma}}\nonumber \\
      &\geq& |x-z| ^{\tilde{\sigma}} (1-200{\tilde{\sigma}}\kappa),\label{gnov2615new,}
\end{eqnarray}
where $\kappa$ will be chosen to be sufficiently small.

 By \eqref{gindnewnew}, \eqref{gnov2615} and \eqref{gnov2615new,}, we have for $|x-z|\geq \frac{M_1}{10}$,
 \begin{eqnarray}
    |G_{\tilde{\Lambda}_1}(x,z)| &= & |G_{J_g}(x,z)| \nonumber\\
    &\leq &M_1^{13gd  M_1^{\tilde{\sigma}} M_0^{\sigma-\tilde{\sigma}}}e^{-  \gamma_0\Gamma_{g}^{\tilde{\sigma}}(z) }\nonumber \\
    &\leq &   e^{-\gamma_0(1-200\kappa\tilde{\sigma}- 300d\kappa\rho\gamma_0^{-1}\frac{\log M_0}{M_0^{1-\rho+\tilde{\sigma}-\sigma}})|x-z|^{\tilde{\sigma}} }.\label{gnov2618}
 \end{eqnarray}
 {\bf Inductions:}
 Define
  \begin{equation}\label{gdec91}
   \gamma_m=\prod_{i=0}^{m-1}\gamma_0(1-C(d)\kappa\tilde{\sigma}- C(d)\kappa \rho\gamma_i^{-1}\frac{\log M_i}{M_i^{1-\rho+\tilde{\sigma}-\sigma}}).
 \end{equation}
 We remind that  $1-\rho+\tilde{\sigma}-\sigma>0$.
 Fix an elementary region $\tilde{\tilde{\Lambda}}_1\in \mathcal{E}_{M_{m}}$ and $\tilde{\tilde{\Lambda}}_1 \subset \tilde{\Lambda}_0$. For any $x\in \tilde{\tilde{\Lambda}}_1$, consider the exhaustion $\{S_j(x)\}_{j=0}^l$ of $\tilde{\tilde{\Lambda}}_1$  at $x$  with width $M_{m-1}$. We say  the annulus  $A_j(x)$ is good, if for any $y\in A_j(x)$, there exists $W(y)\in \mathcal{E}_{M_{m-1}}$ such that
\begin{equation*}
  y\in W(y)\subset A_j(x), \text { dist }(y, A_j(x)\backslash W(y))\geq M_{m-1}/2,
\end{equation*}
and  for $|n-n^\prime|\geq \frac{M_{m-1}}{10}$,
\begin{equation*}
  |(R_{W(y)}A  R_{W(y)})^{-1}(n,n^\prime)|\leq  e^{-\gamma_m|n-n^\prime|^{\tilde{\sigma}}}.
\end{equation*}
Otherwise, we call the annulus bad.
An elementary region $\tilde{\Lambda}_1\subset \tilde{\Lambda}_0$ is called  bad provided for some $x\in \tilde{\Lambda}_1$ the number of bad annuli $\{A_j(x)\}$
exceeds
\begin{equation*}
  B_m:=\kappa \frac{M_{m}}{M_{m-1}}.
\end{equation*}
Otherwise, we call   $\tilde{\tilde{\Lambda}}_1$ good.
Let $\mathcal{F}_m$ be an arbitrary family of pairwise disjoint bad  elementary regions in $\mathcal{E}_{M_m}$ contained in $\tilde{\Lambda}_0$.
By induction, it is easy to see that
\begin{equation*}
 \# \mathcal{F}_m\leq \frac{1}{\kappa ^m}\frac{N^{\varsigma}}{ M_m}.
\end{equation*}
Replace $M_0$, $M_1$, $\gamma_0$, $B_1$ with $M_{m-1}$, $M_m$, $\gamma_{m-1}$, $B_{m}$.
By induction and following the   proof of \eqref{gnov2618}, we have for good elementary regions $\tilde{\tilde{\Lambda}}_1\subset \tilde{\Lambda}_0$ and $\tilde{\tilde{\Lambda}}_1\in \mathcal{E}_{M_m}$,
we have for $|x-z|\geq \frac{M_m}{10}$,
 \begin{eqnarray}
    |G_{\tilde{\Lambda}_1}(x,z)|
    &\leq &   e^{-\gamma_m|x-z|^{\tilde{\sigma}}}.\label{gnov2619}
 \end{eqnarray}
 In order to reach $ \tilde{\Lambda}_0$ after $k$ step, we need
 \begin{equation*}
   M_k= M^{\rho^k}=N,
 \end{equation*}
 hence
 \begin{equation}\label{gnov2621}
   \rho^k\approx \frac{1}{\xi}.
 \end{equation}
 We may modify $\rho$ a little bit at the last  few steps    to ensure that  $k$   is a positive integer.
 However, this issue is rather small.
 Choose $\kappa=N^{-\delta}$  and
 \begin{equation*}
   \delta=-\frac{1}{2}(1-\varsigma)\frac{\log \rho}{\log \xi^{-1}}.
 \end{equation*}
 Direct computations show that
 \begin{equation}\label{gnov2620}
    \# \mathcal{F}_k\leq \frac{1}{\kappa ^k}\frac{N^{\varsigma}}{ M_k}<1.
 \end{equation}
 \eqref{gnov2620} implies that $M_k=N$ is good.
 Therefore, \eqref{gthm1} holds for
 \begin{equation*}
   c_3=\gamma_k,
 \end{equation*}
 where $k$ solves \eqref{gnov2621}.
 Computations show that
 \begin{equation*}
   c_3=c_2-N^{-\vartheta},
 \end{equation*}
 where $\vartheta=\vartheta(\sigma,\tilde{\sigma},\xi,\varsigma)>0$.

\end{proof}

\section{ Proof of Theorem \ref{thmmu2}}
 The proof of  Theorem \ref{thmmu2} is based on matrix-valued Cartan-type  estimates   \cite{bgs,bbook,gs08,jls}.   
 For our purpose,  a new version of Cartan's estimate, which works for non-self-adjoint matrices,  is necessary.
 For convenience,   we include a proof in the Appendix.
\begin{lemma}\label{mcl}
Let $T(x)$ be a   $N\times N$ matrix function of a parameter $x\in[-\delta,\delta]^{J}$ ($J\in\N$) satisfying the following conditions:
\begin{itemize}
\item[(i)] $T(x)$ is real analytic in $x\in [-\delta,\delta]^{J}$ and has a holomorphic extension to
\begin{equation*}
\mathcal{D}_{\delta,\delta_1}=\left\{x=(x_i)_{1\leq i\leq J}\in\mathbb{C}^{J}: \sup_{1\leq i\leq J}|\Re x_i|\leq\delta,\sup_{1\leq i\leq J}|\Im{x_i}|\leq \delta_1\right\}
\end{equation*}
satisfying
\begin{equation}\label{mc1}
\sup_{x\in \mathcal{D}_{\delta,\delta_1}}\|T(x)\|\leq B_1, B_1\geq 1.
\end{equation}
\item[(ii)]  For all $x\in[-\delta,\delta]^{J}$, there is subset $V\subset [1,N]$ with
\begin{equation*}|V|\leq M,\end{equation*}
and
\begin{equation}\label{mc2}
\|(R_{[1,N]\setminus V}T(x)R_{[1,N]\setminus V})^{-1}\|\leq B_2, B_2\geq 1.
\end{equation}
\item[(iii)]
\begin{equation}\label{mc3}
\mathrm{mes}\{x\in[-{\delta}, {\delta}]^{J}: \ \|T^{-1}(x)\|\geq B_3\}\leq 10^{-3J}J^{-J}\delta_1^J(1+B_1)^{-J}(1+B_2)^{-J}.
\end{equation}
Let
\begin{equation}\label{mc4}
0<\epsilon\leq (1+B_1+B_2)^{-10 M}.
\end{equation}
\end{itemize}
Then
\begin{equation}\label{mc5}
\mathrm{mes}\left\{x\in\left[-{\delta}/{2}, {\delta}/{2}\right]^{J}:\  \|T^{-1}(x)\|\geq \epsilon^{-1}\right\}\leq C\delta^Je^{-c\left(\frac{\log \epsilon^{-1}}{M\log(B_1+B_2+B_3)}\right)^{1/J}},
\end{equation}
where $C=C(J), c=c(J)>0$.
\end{lemma}

\begin{proof}[\bf Proof of Theorem \ref{thmmu2}]

Without loss of generality, we  assume   $i=1$. Fix $x_1\in \T^{b_1}$ and $x_1^\neg\in\T^{b-b_1}$.
Recall that $x=(x_1,x_1^\neg)\in \T^b$.

Let $\Lambda=\mathcal{R}\subset [-N_3,N_3]^d$. By making $\mathcal{B}_{\mathcal{R}}(x)$ slightly larger, we have
there exists $\bar{\Lambda}\subset \Lambda$ such that   for any $j\in \Lambda\backslash \bar{\Lambda}$, there  exists $W(j)\in \mathcal{E}_{N_1}$ such that
$W(j)\subset \Lambda\backslash \bar{\Lambda}$, $$ {\rm dist}(j,\Lambda\backslash \bar{\Lambda}\backslash W(j))\geq N_1/2$$ and
\begin{eqnarray}
&&\|G_{W(j)} \|\leq e^{{N_1}^{\sigma}},\label{Apr20}\\
&&|G_{W(j)}(n,n')|\leq e^{-c_2|n-n'|^{\tilde{\sigma}}}\ {\rm for}\ |n-n'|\geq\frac{N_1}{10},\label{Apr21}
\end{eqnarray}
and
\begin{equation}\label{gdec41}
  |\bar{\Lambda}|\leq C(d) L^{1-\delta}N_1^{2d}.
\end{equation}
Indeed, $\bar{\Lambda}$ can be chosen so that
\begin{equation}\label{ggequ2}
  \bar{\Lambda}\subset\{n\in\Z^d: \text{ dist }(n,B_{\mathcal{R}}(x))\leq C(d)N_1.
\end{equation}
Let $\eta= \frac{c_1}{\gamma}$.
 Let $\mathcal{D}$ be the $e^{-\eta N_1}$ neighborhood  of $x_1$ in the complex plane, i,e.,
\begin{equation*}
 \mathcal{D}=\{z\in \mathbb{C}^{b_1}:\  |\Im z|\leq e^{-\eta N_1}, |\Re z-x_1|\leq e^{-\eta N_1}\}.
\end{equation*}

By  the assumption  that $N_3\leq e^{N_1^{\frac{1}{2K_1}}}$, one has for any $y\in \mathcal{D}$,
\begin{equation*}
||x-y||\leq e^{- e^{(\log (2N_3+2))^{K_1}}}
\end{equation*}
and hence (by \eqref{glc1})
\begin{eqnarray}
  | A(x;n,n^\prime)-A(y;n,n^\prime)| &\leq& K||x-y||^{\gamma}\nonumber \\
    &\leq & K e^{-c_1 N_1},\label{gdec43}
\end{eqnarray}
for $n,n^\prime\in[-N_3,N_3]^d$ and large $N_1$.
 By \eqref{Apr20}, \eqref{Apr21}, \eqref{gdec43}, and  standard perturbation arguments, we have  for any $y\in \mathcal{D}$,
and  $j\in\Lambda\backslash\bar\Lambda$,
\begin{eqnarray}
&&\|G_{W(j)}(x_1+y,x_1^\neg)\|\leq 2e^{{N_1}^{\sigma}},\label{Apr22}\\
&&|G_{W(j)}(x_1+y,x_1^\neg;n,n')|\leq 2e^{-c_2|n-n'|^{\tilde{\sigma}}}\ {\rm for}\ |n-n'|\geq\frac{N_1}{10}.\label{Apr23}
\end{eqnarray}
 Substituting   $\Lambda$ with $\Lambda\backslash\bar{\Lambda}$ in Lemma \ref{res1}, one has for any $y\in \mathcal{D}$,
\begin{eqnarray}
\label{bs3}\|G_{\Lambda\setminus \bar\Lambda}(x_1+y,x_1^\neg)\|\leq e^{2{N_1}^{\sigma}}.
\end{eqnarray}
We want to use  Lemma \ref{mcl}.
For this purpose, let  $$T(y)=R_{\Lambda}AR_{\Lambda},J=b_1,\delta=\delta_1=e^{-\eta  N_1}.$$
Now we are in the position to check the assumptions of Lemma \ref{mcl}.
By \eqref{gdec43} and \eqref{GOnew}, one has
$B_1=O(1)$.

Let $V=\bar{\Lambda}$.
By \eqref{gdec41} and \eqref{bs3},  one has
\begin{eqnarray}\label{mb}
 M=|\bar\Lambda|\leq C(d) L^{1-\delta}N_1^{2d},
B_2=e^{2{N_1}^{\sigma}}.
\end{eqnarray}
Applying Lemma \ref{res1} with $M_0=M_1=N_2$ and \eqref{Gstarnov}, one has
\begin{eqnarray*}
 \|T^{-1}(y)\|\leq 4(2N_2+1)^de^{{N_2}^{\sigma}}\leq e^{2{N_2}^\sigma}=:B_3,
\end{eqnarray*}
except  on a   set of $y\in \T^{b_1}$
with measure less than $e^{-N_2^{\zeta}}$.

Since $N_2\geq N_1^{\frac{2}{\zeta}}$,
direct computation shows that
\begin{equation*}
10^{-3b_1}b_1^{-b_1}\delta^{b_1}(1+B_1)^{-b_1}(1+B_2)^{-b_1}\geq e^{-N_2^{\zeta}}.
\end{equation*}
This verifies (iii) in Lemma \ref{mcl}.

Let $\epsilon=e^{-L^{\mu}}$. By (\ref{mb}) and the assumption that $L\geq N^{\frac{2d+b+2}{\mu-1+\delta}}_2$, one has
  $$\epsilon<(1+B_1+B_2)^{-10M}.$$
  Let
\begin{equation*}
  Y=\{y\in\mathcal{D}: \|T^{-1}(y)\|\geq e^{L^{\mu}} \}.
\end{equation*}
By (\ref{mc5}) of Lemma \ref{mcl},
\begin{equation}\label{y0}
\mathrm{mes}( Y)\leq C e^{-c\left(\frac{L^{\mu-1+\delta}}{N_2^{\sigma}N_1^{2d+\sigma}}\right)^{1/b_1}}.
\end{equation}
By covering $\T^{b_1}$ with  balls with radius $e^{-\eta N_1}$, we have
\begin{eqnarray*}
  \mathrm{Leb}( \tilde{X}_{\mathcal{R}}(x_1^\neg)) &\leq& e^{CN_1}  e^{-c\left(\frac{L^{\mu-1+\delta}}{N_2^{\sigma}N_1^{2d+\sigma}}\right)^{1/b_1}}\\
   &\leq &  e^{-\left(\frac{{L}^{\mu-1+\delta}}{N_2^{2d+b+2}}\right)^{1/b_1}},
\end{eqnarray*}
where the second inequality holds by the assumption $L\geq N_2^{\frac{2d+b+2}{\mu-1+\delta}}$. It implies \eqref{Gstar}.

\end{proof}
\section{Proof of Theorems \ref{thmu} and \ref{thmunew}}
\begin{theorem}\label{res2}
	Let $\sigma,\tilde{\sigma},\kappa,s\in(0,1)$ and $\tilde{\sigma}>\kappa$. 
	Assume  ${\rm diam}(\Lambda)\leq 2N+1$.  Let $M_0=(\log N)^{1/s}$.
	Assume
	\begin{equation}\label{gdec93}
	c_2\in (0, (1-5^{-\tilde{\sigma}})c_1].
	\end{equation}
	Suppose that for any $n\in \Lambda $, there exists some  $ W=W(n)\in \mathcal{E}_M$ with $M_0\leq M\leq N^{\kappa}$
	such that
	$n\in W$,  ${\rm dist} (n,\Lambda\backslash W)\geq \frac{M}{2}$, $W\subset \Lambda$  and
	\begin{eqnarray*}
		&& \|G_{W}\|\leq 2e^{{M}^{\sigma}},\\
		&& |G_{W}(n,n')|\leq 2 e^{- c_2|n-n'|^{\tilde{\sigma}}}\  {\mathrm{for} \ |n-n'|\geq \frac{M}{10}}.
	\end{eqnarray*}
	Then
	\begin{equation}\label{gdec19}
	||G_{\Lambda}|| \leq 4(1+2N^{\kappa })^d e^{ N^{\kappa \sigma}},
	\end{equation}
	and
	\begin{eqnarray*}
		|G_{\Lambda}(n,n^\prime)|\leq e^{-\bar{c}|n-n'|^{\tilde{\sigma}}} \text{ for }\  |n-n'|\geq N/10 ,
	\end{eqnarray*}
	where
	\begin{equation}\label{gdec410}
	\bar{c}=c_2-\frac{O(1)}{M_0^{\tilde{\sigma}-s}}-\frac{O(1)}{M_0^{\tilde{\sigma}-\sigma}}-\frac{O(1)}{N^{\tilde{\sigma}-\kappa}}.
	\end{equation}
\end{theorem}
\begin{proof}
	\eqref{gdec19} follows from Lemma \ref{res1} immediately.

	Assume $|n-n^\prime|\geq \frac{N}{10}$.
	Applying \eqref{Greso} with $\Lambda_1=W=W(n)$, one has $ n^\prime\notin W(n)$ and
	\begin{equation*}
	|G_{\Lambda}(n,n')|\leq  \sum_{n_1\in W\atop n_2\in \Lambda\backslash W} e^{-c_1|n_1-n_2|^{\tilde{\sigma}}}|G_{W}(n,n_1)||G_{\Lambda}(n_2,n')|.
	\end{equation*}
	It implies
	\begin{align}
	\nonumber|G_{\Lambda}(n,n')|
	\leq & \sum_{n_1\in W,|n_1-n|\leq \frac{M}{10}-1\atop n_2\in \Lambda\backslash W} e^{-c_1|n_1-n_2|^{\tilde{\sigma}}}|G_{W}(n,n_1)||G_{\Lambda}(n_2,n')|\\
	\nonumber&+ \sum_{n_1\in W,|n_1-n|\geq \frac{M}{10}\atop n_2\in \Lambda\backslash W} e^{-c_1|n_1-n_2|^{\tilde{\sigma}}}|G_{W}(n,n_1)||G_{\Lambda}(n_2,n')|\\
	\leq &\nonumber  \sum_{n_1\in W,|n_1-n|\leq \frac{M}{10}-1\atop n_2\in \Lambda\backslash W} e^{{M}^\sigma}e^{-c_1|n_1-n_2|^{\tilde{\sigma}}}|G_{\Lambda}(n_2,n')|\\
	\nonumber&+ \sum_{n_1\in W,|n_1-n|\geq \frac{M}{10}\atop n_2\in \Lambda\backslash W} e^{-c_1|n_1-n_2|^{\tilde{\sigma}}} e^{-c_2 |n-n_1|^{\tilde{\sigma}}}|G_{\Lambda}(n_2,n')|\\
	\leq &\nonumber \sum_{n_1\in W,|n_1-n|\leq \frac{M}{10}-1\atop n_2\in \Lambda\backslash W} e^{{M}^\sigma}e^{-c_1|n_1-n_2|^{\tilde{\sigma}}}|G_{\Lambda}(n_2,n')|\\
	&+ \sum_{n_1\in W,|n_1-n|\geq \frac{M}{10}\atop n_2\in \Lambda\backslash W} e^{-c_2|n-n_2|^{\tilde{\sigma}}} |G_{\Lambda}(n_2,n')|\nonumber\\
	\leq &\nonumber e^{M^{\sigma}}\sum_{n_1\in W,|n_1-n|\leq \frac{M}{10}-1\atop n_2\in \Lambda\backslash W} e^{-c_2|n-n_2|^{\tilde{\sigma}}}|G_{\Lambda}(n_2,n')|\\
	&+ \sum_{n_1\in W,|n_1-n|\geq \frac{M}{10}\atop n_2\in \Lambda\backslash W} e^{-c_2|n-n_2|^{\tilde{\sigma}}} |G_{\Lambda}(n_2,n')|\nonumber\\
	\leq &\nonumber e^{M^{\sigma}} (2N+1)^{2d}\sup_{n_2\in \Lambda\backslash W}  e^{-c_2|n-n_2|^{\tilde{\sigma}} }|G_{\Lambda}(n_2,n')|\\
	\leq & (2N+1)^{2d}\sup_{n_2\in \Lambda\backslash W}  e^{-(c_2-\frac{O(1)}{M_0^{\tilde{\sigma}-\sigma}})|n-n_2|^{\tilde{\sigma}} }|G_{\Lambda}(n_2,n')|,\label{rsi8}
	\end{align}
	where the  third inequality holds by \eqref{gdec101} and the fourth inequality holds by \eqref{gdec94}.
	
	Iterating \eqref{rsi8} until $|n_2-n'|\leq 4N^{\kappa}$ (but at most $\lfloor\frac{2^{\tilde{\sigma}}|n-n'|^{\tilde{\sigma}}}{M_0^{\tilde{\sigma}}}\rfloor$ times) and applying \eqref{gdec19}, we have for $|n-n'|\geq\frac{N}{10}$,
	\begin{eqnarray*}
		|G_{\Lambda}(n,n')| &\leq & 
		(2N+1)^{ \frac{2^{\tilde{\sigma}}|n-n'|^{\tilde{\sigma}}}{M_0^{\tilde{\sigma}}}} e^{-(c_2-\frac{O(1)}{M_0^{\tilde{\sigma}-\sigma}})(|n-n'|-4N^{\kappa})^{\tilde{\sigma}} } 4(1+2N^{\kappa })^d e^{N^{\kappa\sigma}}\\
		&\leq&e^{ \frac{4|n-n'|^{\tilde{\sigma}}}{M_0^{\tilde{\sigma}}}\log N} e^{-(c_2-\frac{O(1)}{M_0^{\tilde{\sigma}-\sigma}}) (|n-n'|-4N^{\kappa})^{\tilde{\sigma}} } 4(1+2N^{\kappa })^d e^{N^{\kappa\sigma}}\\
		&\leq&  e^{ \frac{4|n-n'|^{\tilde{\sigma}}}{M_0^{\tilde{\sigma}}}M_0^s} e^{-(c_2-\frac{O(1)}{M_0^{\tilde{\sigma}-\sigma}}) (|n-n'|-4N^{\kappa})^{\tilde{\sigma}} } 4(1+2N^{\kappa })^d e^{N^{\kappa\sigma}}\\
		&\leq&   e^{-\bar{c}|n-n'|}  .
	\end{eqnarray*}

\end{proof}
It is easy to see that
the number of  generalized elementary regions in $[-N,N]^d$ with width larger or equal to $N^{\xi}$ is bounded by $N^{C(d)}$, more precisely for any $\xi>0$,
\begin{equation}\label{gbr}
 \# \{\Lambda\subset [-N,N]^d: \Lambda  \subset \mathcal{R}_{L}^{N^{\xi}}\}\leq N^{C(d)}.
\end{equation}

\begin{proof}[\bf Proof of Theorem \ref{thmu}]

Since the Green's function   satisfies properties $P$ with parameters  $(\mu,\zeta,c_2)$ at size $N_2$, we have there exists
$\tilde{X}_{N_2}\subset \mathbb{T}^b$  with
 \begin{equation}\label{gdec30}
 \sup_{1\leq i\leq k,x_i^\neg\in \T^{b-b_i}}\mathrm{Leb}(\tilde{X}_{N_2}(x_i^\neg))\leq N_3^{C(d)}e^{-N_2^{\zeta}},
 \end{equation}
 such that
 \begin{equation*}
  ||G_{m+ Q_{N_2}}(x)||\leq e^{N_2^{\mu}},
\end{equation*}
and for $|n-n^\prime|\geq N_2/10$,
\begin{equation*}
  |G_{m+ Q_{N_2}}(x;n,n^\prime)|\leq  e^{-c_2|n-n^\prime|^{\tilde{\sigma}}},
\end{equation*}
for any $Q_{N_2}\in \mathcal{E}_{N_2}^0$ and  $|m|\leq N_3$.
Indeed, we only need to set
\begin{equation*}
  \tilde{X}_{N_2}=\bigcup_{|m|\leq N_3}X_{N_2}(f^m(x)).
\end{equation*}

By the assumption $N_3\geq N_2^C$ and $N_2\geq N_1^C$ with large $C$ depending on $\varepsilon$, one has
\begin{equation}\label{gdec48}
  N_2\leq N_3^{\varepsilon},N_1\leq N_2^{\varepsilon}.
\end{equation}
Let $\xi=\delta-5\varepsilon$.
Applying \eqref{gdec5} to Theorem \ref{thmmu2}, and by \eqref{gbr} and \eqref{gdec48}, there exists $ {X}_{N_3}\subset [0,1)^b$  such that
\begin{eqnarray}
   \sup_{x_i^\neg\in \T^{b-b_i}}\mathrm{Leb}(X_{N_3}(x_i^\neg)) &\leq&  N_3^{C(d)}e^{-N_3^{\xi(\frac{\sigma-1}{b_i}+\frac{\delta}{b_i})-\varepsilon}}\nonumber \\
   &\leq & e^{-N_3^{\frac{\sigma-1}{b_i}\delta+\frac{\delta^2}{b_i}-\varepsilon}},\label{gdec6u}
\end{eqnarray}
and
for any $x\notin X_{N_3}$, $\mathcal{R}\subset \mathcal{R}_{L}^{N_3^{\xi}}$ with $N_3^{\xi}\leq L\leq N_3$, 
\begin{equation}\label{gnov3011}
||  G_{\mathcal{R}}(x)||\leq e^{L^{\sigma}}.
\end{equation}
 Let $\tilde{\mathcal{F}}$ be any pairwise disjoint elementary regions in $[-N_3,N_3]^d$ with size    $\lfloor N_3^\xi\rfloor $.
By \eqref{gdec5}, it is easy to see that there are at most $N_1^{C(d)}N_3^{1-\delta }=N_3^{1-\delta +\varepsilon}$ in $\tilde{\mathcal{F}}$ will intersect  elementary regions  not in  $SG_{N_1}$.
By Theorem \ref{res2}, any elementary region in $[-N_3,N_3]^d$ with size $\lfloor N_3^\xi\rfloor $, without intersecting any non-$SG_{N_1}$ elementary regions, will satisfy \eqref{ggood}.
It implies \eqref{gnbad} is true for $\varsigma=1-\varepsilon$.
Applying Theorem \ref{thmmul} and \eqref{gnov3011}, we obtain Theorem \ref{thmu}.
Let us  explain where  the bound    $ c_2-N_1^{-\vartheta_1}-N_3^{-\vartheta_2}$ in \eqref{gdec49} is   from.
Since $N_3^{\xi}\leq e^{\xi N_1^{c}}$, one has $s=\frac{11}{10}c$ in Theorem \ref{res2}.  Applying    $M_0=N_1$, $N=N_3^{\xi}$, $\sigma=\mu$ to Theorem \ref{res2}, we obtain the bound $ c_2-O(1)N_1^{-(\tilde{\sigma}-\frac{11}{10}c)}-O(1)N_1^{-(\tilde{\sigma}-\mu)}-N_3^{-\vartheta_2}$.  Theorem \ref{thmmul} will 
only contribute  $N_3^{-\vartheta_2}$. 
\end{proof}
\begin{proof}[\bf Proof of Theorem \ref{thmunew}]
	Fix any $m\in\Z^d$. 
Applying Theorem \ref{thmu} with $\tilde{A}^m$, one has 
there exists a  subset ${X}_{N_3}^m\subset \mathbb{T}^b$  such that
\begin{equation*}
\sup_{1\leq i \leq k,x_i^\neg\in \T^{b-b_i}}\mathrm{Leb}(X_{N_3}^m(x_i^\neg))\leq e^{-{N_3}^{\frac{\sigma-1}{b_i}\delta+\frac{\delta^2}{b_i}-\varepsilon}},
\end{equation*}
	and for any $x\notin  X_{N_3}^m$ and  $Q_{N_3}\in \mathcal{E}_{N_3}^0$,
\begin{eqnarray*}
	|| (R_{Q_{N_3}}\tilde{A}^m(x)R_{Q_{N_3}})^{-1} ||&\leq& e^{N_3^{\sigma}},
\end{eqnarray*}
and  for $|n-n^\prime|\geq \frac{N_3}{10}$, 
\begin{eqnarray*}
	|(R_{Q_{N_3}}\tilde{A}^mR_{Q_{N_3}})^{-1}(x;n,n^\prime)| &\leq& e^{-(c_2-N_1^{-\vartheta_1}-N_{3}^{-\vartheta_2})|n-n^\prime|^{^{\tilde{\sigma}}}}.
\end{eqnarray*}
Let $$ X_{N_3}=\bigcup_{m\in\Z^d} X_{N_3}^m.$$
By  \eqref{comu} and \eqref{gcom} , we have 
$$\sup_{1\leq i \leq k,x_i^\neg\in \T^{b-b_i}}\mathrm{Leb}(X_{N_3}(x_i^\neg))\leq e^{N_3^a}e^{-N_3^{\frac{\sigma-1}{b_i}\delta+\frac{\delta^2}{b_i}-\varepsilon}}\leq e^{-N_3^{\frac{\sigma-1}{b_i}\delta+\frac{\delta^2}{b_i}-\varepsilon}}.$$
\end{proof}
\section{Proof of Theorem \ref{thmholder}}
\begin{proof} 
Once we have the LDT at hand, the modulus of   continuity of the IDS is  standard.
The proof here follows from the corresponding part in \cite{blmp2000,schcmp}.
Let $N=|\log|E_1-E_2||^{\frac{1}{\sigma}-\varepsilon}$. Without loss of generality, assume $E_1<E_2$ and let  $E$ be the center of $[E_1,E_2]$. Therefore,
\begin{equation}\label{edis}
|E_1-E_2|\leq e^{-N^{\sigma+\varepsilon}}.
\end{equation}
By the assumption, there exists a set $X_{N}\subset \T^b$ such that
\begin{equation*}
  {\rm Leb}(X_N)\leq e^{-N^{\zeta}},
\end{equation*}
and for any $x\notin X_N$ and  any $Q_N\in \mathcal{E}_{N}^0$,
\begin{eqnarray*}
  ||G_{Q_N}(E,x)|| &\leq & e^{N^{\sigma}} \\
  |G_{Q_N}(E,x;n,n^\prime)| &\leq& e^{-c|n-n^\prime|^{\tilde{\sigma}}} \text{ for } |n-n^\prime|\geq \frac{N}{10},
\end{eqnarray*}
where $c>0$.  We should mention that $X_N$ depends on $E$.
By the assumption \eqref{gx}, for large $N_1$, one has
\begin{equation*}
  \#\{n\in \Z^d: |n|\leq N_1, f^n(x)\in X_N\}\leq 2(2N_1+1)^d e^{-N^{\zeta}}.
\end{equation*}
Let $\Lambda=[-N_1,N_1]^d$.
By making $ \#\{n\in \Z^d: |n|\leq N_1, f^n(x)\in X_N\}$ slightly larger, we have
there exists $\bar{\Lambda}\subset \Lambda$ such that   for all $j\in \Lambda\backslash \bar{\Lambda}$, there  exists $W(j)\in \mathcal{E}_{N}$ such that
$W(j)\subset \Lambda\backslash \bar{\Lambda}$, $ {\rm dist}(j,\Lambda\backslash \bar{\Lambda}\backslash W(j))\geq N/2$ and
\begin{eqnarray}
&&\|G_{W(j)} \|\leq e^{{N}^{\sigma}},\label{Apr20new}\\
&&|G_{W(j)}(n,n')|\leq e^{-c|n-n'|^{\tilde{\sigma}}}\ {\rm for}\ |n-n'|\geq\frac{N}{10}.\label{Apr21new}
\end{eqnarray}
and
\begin{equation*}
  |\bar{\Lambda}|\leq C(d)N^{2d}(2N_1+1)^d e^{-N^{\zeta}}.
\end{equation*}
Here, $\bar{\Lambda}$ is obtained in a similar way as \eqref{ggequ2}.

 Substituting   $\Lambda$ with $\Lambda\backslash\bar{\Lambda}$ in Lemma \ref{res1}, we have
\begin{eqnarray*}
\|G_{\Lambda\setminus \bar\Lambda}(E,x)\|\leq 4(2N+1)^d e^{{N}^{\sigma}}.
\end{eqnarray*}
By standard perturbation arguments,  we have for any $\tilde{E}\in [E_1,E_2]$,
\begin{eqnarray}
\label{bs3new}\|G_{\Lambda\setminus \bar\Lambda}(\tilde{E},x)\|\leq 8(2N+1)^d e^{{N}^{\sigma}}.
\end{eqnarray}
Denote  by  $\xi_j$, $j=1,2,\cdots,M$, the normalized eigenfunctions
of  $H_{\Lambda}$ with eigenvalues falling into the interval $[E_1,E_2]$. Let  $\xi$ be one of
them with eigenvalue  $\tilde{E}$. By definition,
\begin{equation}\label{gabo}
  R_{\Lambda\backslash\bar{\Lambda}}(H_{\Lambda}-{E}) R_{\Lambda\backslash\bar{\Lambda}}\xi+ R_{\Lambda\backslash\bar{\Lambda}}(H_{\Lambda}-{E}) R_{\bar{\Lambda}}\xi=(\tilde{E}-E)R_{\Lambda\backslash\bar{\Lambda}}\xi.
\end{equation}
Applying $G_{\Lambda\setminus \bar\Lambda}({E},x)$ to \eqref{gabo}, one has
\begin{equation}\label{ghold1}
   R_{\Lambda\backslash\bar{\Lambda}}\xi+ G_{\Lambda\setminus \bar\Lambda}({E},x)R_{\Lambda\backslash\bar{\Lambda}}(H_{\Lambda}-{E}) R_{\bar{\Lambda}}\xi=(\tilde{E}-E)G_{\Lambda\setminus \bar\Lambda}({E},x)R_{\Lambda\backslash\bar{\Lambda}}\xi.
\end{equation}
Denote by $P$ the projection onto the range of $G_{\Lambda\setminus \bar\Lambda}({E},x)R_{\Lambda\backslash\bar{\Lambda}}(H_{\Lambda}-{E}) R_{\bar{\Lambda}}$. Clearly, the
dimension of this range does not exceed $\bar{\Lambda}$. Thus ${\rm rank}(P ) \leq\bar{\Lambda}$.
By \eqref{edis} and \eqref{bs3new}, one has
\begin{equation}\label{ghold2}
|| (\tilde{E}-E)G_{\Lambda\setminus \bar\Lambda}({E},x)R_{\Lambda\backslash\bar{\Lambda}}\xi||\leq \frac{1}{100}||\xi||.
\end{equation}
Applying $I-P$ to \eqref{ghold1} and by \eqref{ghold2}, we have
\begin{equation}\label{ghold3}
||  R_{\Lambda\backslash\bar{\Lambda}}\xi-P  R_{\Lambda\backslash\bar{\Lambda}}\xi ||\leq  \frac{1}{100}||\xi||.
\end{equation}
Applying  \eqref{ghold3} to each $\xi_j$,
we have
\begin{eqnarray*}
  M &=&\sum_{j=1}^N ||\xi_j||^2 \\
   &\leq&\frac{M}{2}+4\sum_{j=1}^M||PR_{\Lambda\backslash\bar{\Lambda}}\xi_j||^2+ 2\sum_{j=1}^M||R_{\bar{\Lambda}}\xi_j||^2\\
    &\leq&\frac{M}{2}+4{\rm Trace}(PR_{\Lambda\backslash\bar{\Lambda}})+ 2{\rm Trace}(R_{\bar{\Lambda}})\\
    &\leq& \frac{M}{2}+ 6|\bar{\Lambda}|\\
   &\leq& \frac{M}{2}+ C(d)N^{2d}(2N_1+1)^d e^{-N^{\zeta}}.
\end{eqnarray*}
Therefore,
\begin{equation*}
  M\leq C(d)N^{2d}(2N_1+1)^d e^{-N^{\zeta}}.
\end{equation*}
It implies
\begin{equation*}
  k(x,E_1,E_2)\leq C(d)N^{2d} e^{-N^{\zeta}}\leq e^{-(\log\frac{1}{E_2-E_1})^{\frac{\zeta}{\sigma}-\varepsilon}}.
\end{equation*}
\end{proof}

\section{The discrepancy and semi-algebraic sets}

\subsection{ Discrepancy }
Let $ \vec{x}_1,..., \vec{x}_N\in [0,1)^b$ and $\mathcal{S}\subset [0,1)^b$.
Let $A(\mathcal{S}; \{\vec{x}_n\}_{n=1}^N)$ be the number of  $\vec{x}_n$ ($1\leq n\leq N$)
such that $\vec{x}_n\in \mathcal{S}$.
We define
the discrepancy of the sequence $\{\vec{x}_n\}_{n=1}^N$ by
\begin{align}\label{counting}
D_N(\{\vec{x}_n\}_{n=1}^N)=\sup_{\mathcal{S}\in \mathcal{C}}\left|\frac{A(\mathcal{S}; \{\vec{x}_n\}_{n=1}^N)}{N}-\mathrm{Leb}(S)\right|,
\end{align}
where $ \mathcal{C}$ is the family of all  intervals in $[0,1)^b$, namely $\mathcal{S}$ has the form of $$\mathcal{S}=[\varrho_1,\beta_1]\times [\varrho_2,\beta_2]\times \cdots \times [\varrho_b,\beta_b]$$ with $0\leq \varrho_n<\beta_n<1$, $n=1,2,\cdots,b$.
Let  $\alpha=(\alpha_1,\alpha_2,\cdots,\alpha_b)\in [0,1)^b$. The $b$-dimensional sequence  $\vec{x}_n=(n\alpha_1,n\alpha_2,\cdots,n\alpha_b) \mod \Z^b$ ($n\alpha$ for short), $n=1,2,\cdots$, is called the Kronecker sequence.  We denote by the discrepancy of $ \{n\alpha\}_{n=1}^N$, $D_N(\alpha)$.
The following Lemmas are well known.
\begin{lemma}\cite{dt97}\label{ledissh}
Assume $\alpha\in {\rm DC}(\kappa,\tau) $. Then
\begin{equation*}
  D_N(\alpha)\leq C(b,\kappa,\tau)N^{-\frac{1}{\kappa}}(\log N)^2.
\end{equation*}
\end{lemma}
\begin{lemma}\cite{s64}\label{ledissh1}
For almost every $\alpha$,   we have
\begin{equation*}
  D_N(\alpha)\leq C(\alpha)N^{-1}(\log N)^{b+2}.
\end{equation*}
\end{lemma}
Let $f$: $\T^b\rightarrow \T^b$ be defined as follows
\begin{equation*}
T(y_1,y_2,...,y_b)=(y_1+\alpha, y_2+y_1,...,y_b+y_{b-1}).
\end{equation*}
Let $T^n$  be the $n$th iteration of $T$ and $\vec{Y}_n=T^n(y_1,...,y_b)$.
\begin{lemma}\label{ledissk}
Assume $\alpha\in {\rm DC}(\kappa,\tau) $. Then for any $\varepsilon>0$,
\begin{equation*}
  D_N(\{\vec{Y}_n\}_{n=1}^N)\leq C(b,\kappa,\tau,\varepsilon)N^{-\frac{1}{2^{b-1}\kappa}+\varepsilon}.
\end{equation*}
\end{lemma}
\begin{remark}
	Lemma \ref{ledissk}  follows from  the Erd\H{o}s-Tur\'an inequality (see Corollary 1.1 in p.8 of \cite{hugh}) and the Weyl's method  (Theorem 2 in p.41 of \cite{hugh}).
\end{remark}
The Erd\H{o}s-Tur\'an inequality and Weyl's method  also imply
\begin{lemma} \label{ledisskp}
Assume $\alpha\in {\rm DC}(\kappa,\tau) $. Let $ Y_n=P_b(T^n(y_1,...,y_b))$, where $P_b$ is the $b$th coordinate projection.
 Then for any $\varepsilon>0$,
\begin{equation*}
  D_N(\{{Y}_n\}_{n=1}^N)\leq C(b,\kappa,\tau,\varepsilon)N^{-\frac{1}{2^{b-1}\kappa}+\varepsilon}.
\end{equation*}
\end{lemma}

\subsection{Semi-algebraic sets}
A set $\mathcal{S}\subset \mathbb{R}^n$ is called a semi-algebraic set if it is a finite union of sets defined by a finite number of polynomial equalities and inequalities. More precisely, let $\{P_1,\cdots,P_s\}\subset\mathbb{R}[x_1,\cdots,x_n]$ be a family of real polynomials whose degrees are bounded by $d$. A (closed) semi-algebraic set $\mathcal{S}$ is given by an expression
\begin{equation}\label{smd}
\mathcal{S}=\bigcup\limits_{j}\bigcap\limits_{\ell\in\mathcal{L}_j}\left\{x\in\mathbb{R}^n: \ P_{\ell}(x)\varsigma_{j\ell}0\right\},
\end{equation}
where $\mathcal{L}_j\subset\{1,\cdots,s\}$ and $\varsigma_{j\ell}\in\{\geq,\leq,=\}$. Then we say that $\mathcal{S}$ has degree at most $sd$. In fact, the degree of $\mathcal{S}$ which is denoted by $\deg(\mathcal{S})$, means the  smallest $sd$ over all representations as in (\ref{smd}).

The following  lemma is a special case appearing \cite{ba}. It is restated in \cite{bbook}.
\begin{lemma}\cite[Theorem 9.3]{bbook} \cite[Theorem 1]{ba}\label{lediscom}
Let $\mathcal{S}\subset [0,1]^n$ be a semi-algebraic  set of degree $B$. Then  the number of connected components of
$\mathcal{S}$ does not exceed $(1+B)^{C(n)}$.
\end{lemma}
The following  lemma  follows from  the Yomdin-Gromov triangulation theorem \cite{gro,yom}, which has been stated in \cite{bbook}. 
We refer readers to  \cite{BiN}  and references therein for the complete proof of the Yomdin-Gromov triangulation theorem.
\begin{lemma}\cite[Corollary 9.6]{bbook}\label{ledis}
Let $\mathcal{S}\subset [0,1]^n$ be a semi-algebraic set of degree $B$. Let $\epsilon>0$ be a small number and ${\rm Leb}(\mathcal{S})\leq \epsilon^n$. Then $\mathcal{S}$
can be covered by a family of $\epsilon$-balls with total number less than $\frac{(1+B)^{C(n)} }{\epsilon^{n-1}}$.
\end{lemma}

\begin{theorem}\label{thmdis}
Assume that  the discrepancy of the sequence $\{\vec{x}_j\}_{j=1}^N$ satisfies
\begin{equation*}
  D_{N}(\{\vec{x}_j\}_{j=1}^N)\leq N^{-\varsigma},
\end{equation*}
for some $\varsigma>0$.
Let $S\subset [0,1]^n$ be a semi-algebraic set  with degree less than $B$.  Suppose
\begin{equation*}
  {\rm Leb}(\mathcal{S})\leq N^{-\varsigma}.
\end{equation*}
Then
\begin{equation*}
   A(\mathcal{S}; \{\vec{x}_j\}_{j=1}^N)\leq (1+B)^{C(n)} N^{1-\frac{\varsigma}{n}}.
\end{equation*}
\end{theorem}
\begin{proof}
Let $\epsilon=N^{-\frac{\varsigma}{n}}$. By Lemma \ref{ledis}, ${S}$ can be covered, at most $\frac{(1+B)^C }{\epsilon^{n-1}}$,  $\epsilon $-balls.
Pick  one $\epsilon $-ball, say $J$.  By the fact   $D_{N}(\{\vec{x}_j\}_{j=1}^{N})\leq N^{-\varsigma}$, one has
\begin{equation*}
  A(J; \{\vec{x}_j\}_{j=1}^N)\leq CN \epsilon^n+N^{1-\varsigma}\leq CN^{1-\varsigma},
\end{equation*}
where $C$ depends on the dimension $n$.
Since there are at most $\frac{(1+B)^C }{\epsilon^{n-1}}$ balls, we have
\begin{eqnarray*}
   A(\mathcal{S}; \{\vec{x}_j\}_{j=1}^N) &\leq& (1+B)^C \frac{1}{\epsilon^{n-1}} N^{1-\varsigma} \\
   &=& (1+B)^C N^{\frac{n-1}{n}\varsigma} N^{1-\varsigma} \\
   &=&  (1+B)^C N^{1-\frac{\varsigma}{n}}.
\end{eqnarray*}
\end{proof}
\begin{remark}\label{redec82}
\begin{itemize}
\item Theorem \ref{thmdis} says that there is a factor  $b$ loss (referred to as dimension loss) when passing discrepancy from intervals to semi-algebraic sets.   The
dimension loss is not surprising. For example, there is also a dimension loss passing the discrepancy to the isotropic discrepancy \cite[Theorem 1.6]{kui74}.
  \item  The   proof  of Theorem \ref{thmdis} is taken from Bourgain \cite{bbook}, where no explicit bounds are given.
\end{itemize}
\end{remark}
For a set $S\subset [0,1)^2$, denote by $l(S)$ the length of the longest line segment
contained in $S$.
\begin{lemma}\cite[Theorem 5.1]{bk}\label{lebk}
Assume $\alpha_1\in {\rm DC}(\kappa,\tau)$ and $\alpha_2\in {\rm DC}(\kappa,\tau)$.
Let $S\subset [0,1)^2$ be a semi-algebraic set with degree less than $B$ and
\begin{equation*}
  l(S)\leq \frac{1}{2}\min_{1\leq |k|\leq 2N}||k\alpha||.
\end{equation*}
Then
\begin{equation*}
  \#\{k=(k_1,k_2)\in \Z^2: |k|\leq N, (k_1\alpha_1,k_2\alpha_2)\in S \mod \Z^2\}
\end{equation*}
\begin{equation}\label{gbkb}
\leq (1+B)^{C(d)} C(\kappa,\tau)N^{3\kappa-\frac{9}{4}}.
\end{equation}
\end{lemma}

\section{Proof of all the results in Section \ref{sapp}}

Applying    Theorem \ref{thmholder} with $\sigma=1-\varepsilon$, 
 Theorem \ref{thmapp1'}  follows from Theorem \ref{thmapp1},
Theorem \ref{thmapp4'} follows from Theorem \ref{thmapp4}, Theorem \ref{thmapp2'} follows from Theorem \ref{thmapp2},  Theorem \ref{thmapp5'} follows from Theorem \ref{thmapp5}, Theorem \ref{thmapp7'} follows from Theorem \ref{thmapp7} and Theorem \ref{thmapp3'} follows from Theorem \ref{thmapp3}.

Applying   strong Diophantine frequencies to  Theorems   \ref{thmapp5'} and  \ref{thmapp3'}, we obtain Corollaries
  \ref{coroapp5} and \ref{coroapp3}.

 With large deviation theorems \ref{thmapp4} and   \ref{thmapp2}  at hand, the proof  of Theorems \ref{thmapp4new} and    \ref{thmapp2''}  is rather standard. We refer the readers to \cite[Section 3]{gafa},  \cite[Section 6]{bgs} and \cite[Chapter XV]{bbook} for details. We note that the only difference is that the degree of semi-algebraic sets is at most $e^{(\log N)^C}$ in our cases, not $N^C$.

By the discussion above, in order to prove all the results in Section \ref{sapp}, it suffices to prove Theorems 
\ref{thmapp1}, \ref{coroapp1}, \ref{thmapp4},   \ref{coroapp4}, \ref{thmapp2},    \ref{thmapp5}, \ref{thmapp7},    \ref{thmapp3} and Corollary \ref{coroapp5b}.

In this section, $C$($c$) is always a  large (small) constant. It may change even in the same formula.
\begin{lemma}\cite[Prop.7.19]{bbook}\label{leinitial}
	Let $H(x)$ be given by \eqref{opapp1} and the Lyapunov exponent is given by \eqref{G21}.
Suppose $L(E)>0$. Then  there exist $0<\sigma<1$ and $\zeta>0$ such that
for large $N$, there exists $X_N\subset \T^b$ such that ${\rm Leb}(X_N)\leq e^{-N^\zeta}$ and
 for $x\notin X_N$, one of the intervals
\begin{equation*}
  \Lambda=[1,N];[1,N-1];[2,N];[2,N-1]
\end{equation*}
will satisfy
\begin{equation*}
  |G_{\Lambda}(n_1,n_2)|\leq e^{-  L(E)|n_1-n_2|+N^\sigma}.
\end{equation*}
\end{lemma}
\begin{proof}[\bf Proof of Theorem \ref{thmapp1}]
By Lemma \ref{leinitial},
 there exist $0<\sigma_1<1$ and $\zeta_1>0$ such that
for any large $N_1$, there exists $X_{N_1}\subset \T^b$ such that ${\rm Leb}(X_{N_1})\leq e^{-N_1^{\zeta_1}}$ and
 for $x\notin X_{N_1}$, one of the intervals
\begin{equation}\label{gnov303}
  \Lambda(N_1)=[1,N_1];[1,N_1-1];[2,N_1];[2,N_1-1]
\end{equation}
will satisfy
\begin{equation}\label{gnov301}
  |G_{\Lambda(N_1)}(n_1,n_2)|\leq e^{-  L(E)|n_1-n_2|+N_1^{\sigma_1}}.
\end{equation}
By approximating the analytic function with trigonometric polynomials given by  \eqref{glc1news1} and using Taylor expansions,
we can further  assume that $ X_{N_1}$ is a semi-algebraic set with degree less than $e^{(\log N_1)^C}$. This argument is quite standard. We refer to \cite{bbook} for details.
By Lemma \ref{ledissh} and Theorem \ref{thmdis},
for any $e^{(\log N_1)^C}\leq N_3\leq e^{N_1^c}$,
\begin{equation*}
  A(X_{N_1}; \{n\omega\}_{n=1}^{N_3}) \leq N_3^{1-\frac{1}{b\kappa}+\varepsilon}.
\end{equation*}
Let $N_2=N_3^{\frac{1}{C}}$. Applying \eqref{gnov301} to $N_2$, one has
\begin{equation}\label{gnov302}
  |G_{\Lambda(N_2)}(n_1,n_2)|\leq e^{-  L(E)|n_1-n_2|+N_2^{\sigma_1}},
\end{equation}
except for  a set of $x$ with measure less than $e^{-N_2^{\sigma_1}}$.
Now Theorem \ref{thmapp1} follows from Theorem \ref{thmu}.
We should mention that the elementary region is  $ [-N_1,N_1]$ in Theorem \ref{thmu} which is slightly different from \eqref{gnov303}. However,  the same statement is true.
\end{proof}
\begin{proof}[\bf Proof of Theorem \ref{thmapp7}]
The proof of Theorem \ref{thmapp7} is similar to that of Theorem \ref{thmapp1}.
The difference is that instead of Lemma \ref{leinitial}, we need to use the corresponding statements in p.3575 \cite{taojde} for  initial scales.
We also need to use Lemma \ref{ledisskp} instead of Lemma \ref{ledissk}. 
\end{proof}

\begin{proof}[\bf Proof of Theorem \ref{thmapp2}]
Let   $N_2=e^{N_1^c}$.  Assume the Green's function in Theorem  \ref{thmapp2} satisfies properties $P$ with parameters  $(\mu,\zeta,c_2)$ at sizes $N_1$
and $N_2$.
 Let $N_3=N_2^C$. 
We can  assume that $ X_{N_1}$ is a semi-algebraic set with degree less than $e^{(\log N_1)^C}$.
By Lemma \ref{lediscom}, $ X_{N_1}$ is consisted  of at most $e^{(\log N_1)^C}$ intervals with measure less than $e^{-N_1^\zeta}$. Let $I$ be one of the intervals.
Since $\omega$ satisfies Diophantine condition, for any $x\in\T$, there is at most one $n\in \Z^d$ with  $|n|\leq N_3$ such that $x+n\omega\mod \Z\in I$.
Therefore,
\begin{equation}\label{gnov3010}
  A(X_{N_1}; \{n\omega\}_{n=1}^{N_3}) \leq e^{(\log N_1)^C}\leq N_3^{\varepsilon}.
\end{equation}
By Theorem \ref{thmu}, we have the Green's function satisfies
properties $P$ with parameters     $(\sigma,\sigma-\varepsilon,c_2-N_3^{-\vartheta})$  at size $N_3$.
Standard Neumann series expansion ensures that for  any large $N_0$, there exists $\lambda_0$ such that for any $\lambda>\lambda_0$, the Green's functions have properties $P$ with parameters  $(\sigma,\sigma-\varepsilon,\frac{4}{5}c_1)$ at all sizes
smaller than $N_0$ \cite[Theorem 4.3]{jls}. Now Theorem   \ref{thmapp2} follows by standard induction. See pages 15 and 16 in \cite{jls} for details.

\end{proof}

\begin{proof}[\bf Proof of Theorem  \ref{thmapp4}]
Fix $N_1$. Let $N_2=e^{N_1^c} $ and $N_3=N_2^C$. Assume the Green's function in Theorem  \ref{thmapp4}    satisfies properties $P$ with parameters  $(\mu,\zeta,c_2)$ at sizes $N_1$
and $N_2$.  We can again  assume that $ X_{N_1}$ is a semi-algebraic set with degree less than  $e^{(\log N_1)^C}$.
By Lemma \ref{ledissh} and Theorem \ref{thmdis},
\begin{equation}\label{gdec31}
  A(X_{N_1}; \{n\omega\}_{n=1}^{N_3}) \leq N_3^{1-\frac{1}{b\kappa}+\varepsilon}.
\end{equation}
By Theorem \ref{thmu}, we have the Green's function satisfies
properties $P$ with parameters     $$\left(\sigma,\frac{\sigma-1}{b^2\kappa}+\frac{1}{b^3\kappa^2}-\varepsilon,c_2-N_3^{-\vartheta}\right)$$  at size $N_3$.
As the arguments at  the end  of proof of Theorem  \ref{thmapp2}, large $\lambda$ will ensure the initial scales and hence Theorem  \ref{thmapp4} follows by induction.
\end{proof}

\begin{proof}[\bf Proof of Theorems  \ref{coroapp1} and   \ref{coroapp4}]
The proof of Theorems   \ref{coroapp1} and   \ref{coroapp4} closely follow that of Theorems \ref{thmapp1'} and   \ref{thmapp4'}.
The difference is that we need to use Lemma \ref{ledissh1} instead of  Lemma \ref{ledissh}.
\end{proof}

\begin{proof}[\bf Proof of Theorem  \ref{thmapp5}]
Replacing Lemma \ref{ledissh} with Lemma \ref{ledissk},   Theorem  \ref{thmapp5} follows Theorem \ref{thmapp4}.
\end{proof}

\begin{proof}[\bf Proof of Corollary \ref{coroapp5b}]
By formula (3.53) in \cite{bgscmp},  one has for almost every $\alpha$,
\begin{equation}\label{gdec31new}
  A(X_{N_1}; \{n\omega\}_{n=1}^{N_3}) \leq N_3^{1-\frac{1}{3}+\varepsilon}.
\end{equation}
Let $\delta=1/3-\varepsilon$. Applying $\tilde{\sigma}=1$, $\sigma=1-\varepsilon$ and $b_i=2$ in Theorem \ref{thmu} and then Theorem \ref{thmholder}, we obtain 
Corollary \ref{coroapp5b}. Indeed, $1/18$ comes from $(1/3)^2/b$.
	
\end{proof}

\begin{proof}[\bf Proof of Theorem  \ref{thmapp3}]
The proof of Theorem  \ref{thmapp3} is similar to that of
 Theorems \ref{thmapp2} and \ref{thmapp4}. We only  point  out the modifications.
 \begin{itemize}
   \item The induction goes in the following way. The semi-algebraic set $X_N$ intersecting with any line segments contained in $[0,1)^2$ has Lebesgue measure at most
   $e^{-N^{\zeta}}$. The assumption that $v$ is not constant on any line segments  ensure the initial scales.
   \item  Replace \eqref{gnov3010} or \eqref{gdec31} with \eqref{gbkb}.
   \item  Since  the induction is based on semi-algebraic sets  only on line segments, the Cartan's estimate will not lead to dimension loss. In other words, when \eqref{gdec3u} is used to do the induction,   $b_i=1$.
    \end{itemize}
\end{proof}
\begin{remark}\label{rec}
\begin{enumerate}
  \item The calculation of the bound in  Theorem  \ref{thmapp3} goes in the following way. By \eqref{gbkb}, the sublinear bound is
   \begin{equation*}
     3\kappa-\frac{9}{4}=1-\delta, \text{ where } \delta=\frac{13}{4}-3\kappa.
   \end{equation*}
    Therefore, the  bound in \eqref{gdec3u} becomes ($b_i=1$)
    \begin{eqnarray*}
      \frac{\sigma-1}{b_i}\delta+\frac{\delta^2}{b_i} &=& (\sigma-1)\delta+\delta^2\\
       &=& (\sigma-1)\left(\frac{13}{4}-3\kappa\right)+\left(\frac{13}{4}-3\kappa\right)^2.
    \end{eqnarray*}
  \item  The induction   of  Theorem  \ref{thmapp3}  follows the corresponding parts in \cite{bk}. Our quantitative approaches developed in the paper allow us to obtain the explicit bound.
\end{enumerate}
\end{remark}

  \appendix
\section{Cartan's estimates for non-self-adjoint matrices}
In the following, we will prove the several variables matrix-valued  Cartan estimate (Lemma \ref{mcl}). The proof is similar to that in \cite{bbook,bkick,jls,gafa,bgs}. The improvement is that  we do not assume the matrix is self-adjoint.
\begin{lemma}\label{sl}
	Let $T$ be the matrix
	\begin{equation*}
	T=\left(
	\begin{array}{cc}
	T_1&T_2\\
	{T}_3&T_4
	\end{array}
	\right),
	\end{equation*}
	where $T_1$ is an invertible $n\times n$ matrix , $T_2$ is an $n\times k$ matrix,  $T_3$ is a $k\times n$ matrix, and $T_4$ is a $k\times k$ matrix. Let
	$$S=T_4- {T}_3T_1^{-1} T_2.$$
	Then $T$ is invertible if and only if $S$ is invertible, and
	\begin{equation}\label{sl1}
	\|S^{-1}\|\leq \|T^{-1}\|\leq C(1+||T_2||)(1+||T_3||)(1+\|T_1^{-1}\|)^2(1+\|S^{-1}\|),
	\end{equation}
	where $C$  is an absolute constant.
\end{lemma}
\begin{proof}
	It is easy to check that
	\begin{equation}\label{appaug1}
	T=\left(
	\begin{array}{cc}
	T_1&T_2\\
	{T}_3&T_4
	\end{array}
	\right)=\left(
	\begin{array}{cc}
	I&0\\
	{T}_3T_1^{-1}&I
	\end{array}
	\right)\left(
	\begin{array}{cc}
	I&T_2\\
	0&S
	\end{array}
	\right)\left(
	\begin{array}{cc}
	T_1&0\\
	0&I
	\end{array}
	\right).
	\end{equation}
	It implies $T$ is invertible if and only if $S$ is invertible.  
	By  \eqref{appaug1}, one has
	\begin{eqnarray}
		T^{-1}&=&  \left(
		\begin{array}{cc}
			T_1&0\\
			0&I
		\end{array}
		\right)^{-1}\left(
		\begin{array}{cc}
			I& T_2 \\
			0& S
		\end{array}
		\right)^{-1}  \left(
		\begin{array}{cc}
			I&0\\
		 {T}_3T_1^{-1}&I
		\end{array}
		\right)^{-1} \nonumber \\
		&=& \left(
		\begin{array}{cc} 
			T_1^{-1}&0\\
			0&I
		\end{array}
		\right)\left(
		\begin{array}{cc}
			I&-T_2S^{-1}\\
			0& S^{-1}
		\end{array}
		\right)   \left(
		\begin{array}{cc}
			I&0\\
			-{T}_3T_1^{-1}&I
		\end{array}
		\right)\label{app1}\\
		&=&
		\left(
		\begin{array}{cc}
			\star&\star\\
			\star& S^{-1}
		\end{array}
		\right).\label{app2}
	\end{eqnarray}
Now 
the second inequality of \eqref{sl1} follows from \eqref{app1} and the first one follows from \eqref{app2}.
\end{proof}
 Denote by $\mathcal{D}(z,r)$  the standard disk on $\mathbb{C}$ of center  $z$ and radius  $r>0$.

\begin{lemma}\cite[Lemma 2.15]{gs}\label{svcl}
	Let $f(z_1,\cdots,z_J)$ be an analytic function defined in a ploydisk $\mathcal{P}=\prod\limits_{1\leq i\leq J}\mathcal{D}(z_{i,0},1/2)$ and $\phi=\log|f|$. Let $\sup\limits_{\underline{z}\in \mathcal{P}}\phi(\underline{z})\leq M,m\leq \phi(\underline{z}_0)$, $\underline{z}_0=(z_{1,0},\cdots,z_{J,0})$. Given sufficiently large $F $, there exists a set $\mathcal{B}\subset\mathcal{P}$  such that
	\begin{equation}\label{cal1}
	\phi(\underline{z})>M-C(J)F(M-m), \text{ for any }\ \underline{z}\in \prod\limits_{1\leq i\leq J}\mathcal{D}(z_{i,0},1/4)\setminus \mathcal{B},
	\end{equation}
	and
	\begin{equation}\label{cal2}
	\mathrm{mes}(\mathcal{B}\cap\mathbb{R}^J)\leq  C(J)e^{-F^{1/J}}.
	\end{equation}
\end{lemma}

\begin{proof}[\bf{Proof of Lemma \ref{mcl}}]
	The proof is similar to that of Lemma 3.4 in \cite{bbook}.
		In the following proof, $C=C(J)$  and $c=c(J)$.
	
	Let
	$$\mu=10^{-2}{J^{-1}}\delta_1(1+B_1)^{-1}(1+B_2)^{-1}.$$ Fix
	$$x_0\in\left[-\delta/2, \delta/2\right]^{J}$$
	and consider $T(z)$ with $|z-x_0|=\sup\limits_{1\leq i\leq J}|z_i-x_{0,i}|<\mu$.  Thanks to Cauchy's estimate and (\ref{mc1}), one obtains for $|z-x_0|<\mu$,
	$$\|{\partial_{z_i} T(z)}\|\leq \frac{4 B_1}{\delta_1},i=1,2,\cdots, J,$$
	which implies
	$$\|T(z)-T(x_0)\|\leq \frac{4JB_1\mu}{\delta_1}\leq 25^{-1}(1+B_2)^{-1}.$$
	From the assumption (ii) of Lemma \ref{mcl}, we can find $V=V(x_0)$ so that $|V|=\tilde{M}\leq M$ and \eqref{mc2} is satisfied. Denote by $V^c=[1,N]\setminus V$. Thus using the standard Neumann series argument and (\ref{mc2}), one has
	\begin{equation}\label{acb}
	\|(R_{V^c}T(z)R_{V^c})^{-1}\|\leq 2B_2\ \mathrm{for}\  |z-x_0|<\mu.\end{equation}
	We define for $|z-x_0|<\mu$ the analytic   function
	\begin{equation}\label{scf}
	S(z)=R_{V}T(z)R_{V}-R_{V}T(z)R_{V^c}(R_{V^c}T(z)R_{V^c})^{-1}R_{V^c}T(z)R_{V}.
	\end{equation}
	Then by (\ref{acb}) and (\ref{scf}), we have
	\begin{equation}\label{sb}
	\|S(z)\|\leq 3B_1^2B_2.
	\end{equation}
	Recalling Lemma \ref{sl}, if $S(z)$ is invertible, so is $T(z)$ and by (\ref{sl1}),
	\begin{equation}\label{sib}
	\|S^{-1}(z)\|\leq C\|T^{-1}(z)\|\leq CB_1^2B_2^2(1+\|S^{-1}(z)\|).
	\end{equation}
	For $x\in\mathbb{R}^{J}$, 
	one has
	\begin{equation}\label{dets}
	||S(x)||^{\tilde{M}}\geq |\det S(x)|.
	\end{equation}
		Let $\lambda=\min \{|\tilde{\lambda}|: \tilde{\lambda}\in \sigma(S(x)) \}$.
	We have
	\begin{align}
|\det S(x)|&\geq \lambda ^{\tilde{M}}
\nonumber\\
&\geq \|S^{-1}(x)\|^{-\tilde{M}}.\label{dets'}
	\end{align}
	By Cramer's rule, one has every  entry of $S^{-1}(x)$ is bounded by
	$$\frac{||S(x)||^{{\tilde{M}}-1}}{|\det S(x)|}$$
	and hence  (by (\ref{sb}))
		\begin{equation}\label{sib1}
	\|S^{-1}(x)\|\leq \frac{\tilde{M}(3B_1^2B_2)^{\tilde{M}}}{|\det S(x)|}.
	\end{equation}

	Let
	$$\phi(z)=\log|\det S(x_0+\mu z)|,\  |z|<1.$$
	Then by (\ref{dets}) and (\ref{sb}),
	\begin{equation}\label{pub}
	\sup_{|z|<1}\phi(z)\leq C \tilde{M}\log (B_1+B_2).
	\end{equation}
	By (\ref{mc3}) and the definition of $\mu$, there is some $x_1$ with $|x_0-x_1|<\mu/10$  such that
	\begin{equation}
	\|T^{-1}(x_1)\|\leq B_3.
	\end{equation}
	Hence by (\ref{sib}), $\|S^{-1}(x_1)\|\leq CB_3$, and from (\ref{dets'}),
	\begin{equation}\label{plb}
	\phi(a)\geq -C\tilde{M}\log B_3,
	\end{equation}
	where $a=\frac{x_1-x_0}{\mu}$, so $ |a|<1/10$.      Let
	\begin{equation*}
	\mathcal{P}=\prod_{1\leq i\leq J}\mathcal{D}(a_i,{1}/{2}).
	\end{equation*}
	Therefore, one has
	\begin{equation*}
	\sup_{z\in\mathcal{P}}\phi(z)\leq C \tilde{M}\log (B_1+B_2), \phi(a)\geq- C\tilde{M}\log B_3.
	\end{equation*}
	Applying Lemma \ref{svcl} and recalling (\ref{cal1}), (\ref{cal2}), for any $F\gg1$, there is some set $\mathcal{B}\subset \prod\limits_{1\leq i\leq J}\mathcal{D}(a_i,{1}/{4})$ with
	\begin{equation}\label{plbb}
	\phi(z)\geq -CF \tilde{M}\log(B_1+B_2+B_3)\ \mathrm{for}\ z\in \prod\limits_{1\leq i\leq J}\mathcal{D}(a_i,{1}/{4})\setminus \mathcal{B},
	\end{equation}
	and
	\begin{equation}\label{bm}
	\mathrm{mes}(\mathcal{B}\cap\mathbb{R}^{J})\leq Ce^{-F^{1/J}}.
	\end{equation}
	For $0<\epsilon<1$, let
	\begin{equation*}
	F =\frac{-c\log \epsilon}{ \tilde{M}\log(B_1+B_2+B_3)}.
	\end{equation*}
	Then by (\ref{plbb}) and (\ref{bm}),
	\begin{eqnarray}
	\nonumber&&\mathrm{mes}\left\{x\in\mathbb{R}^{J}: \ |x-x_1|<\mu/4\ \mathrm{and}\ |\det (S(x))|\leq \epsilon\right\}\\
	\nonumber&&\ \ \ \ \ \  = \mu^{{J}}\mathrm{mes}\left\{x\in\mathbb{R}^{J}:\  |x-a|<1/4\ \mathrm{and}\ \phi(x)\leq \log \epsilon\right\}\\
	\nonumber&&\ \ \ \ \ \ \leq C\mu^{{J}} e^{-F^{1/J}}.
	\end{eqnarray}
	Since  $|x_0-x_1|<\mu/10$, we have
	\begin{equation}\label{x0b}
	\mathrm{mes}\left\{x\in\mathbb{R}^{J}:\  |x-x_0|<\mu/8\ \mathrm{and}\ |\det(S(x))|\leq \epsilon\right\}\leq C \mu^{{J}} e^{-c\left(\frac{\log \epsilon^{-1}}{\tilde{M} \log(B_1+B_2+B_3)}\right)^{1/J}}.
	\end{equation}
	Recalling (\ref{sib}), (\ref{sib1}) and (\ref{mc4}), one has for $|x-x_0|<\mu/8$ and $|\det S(x)|\geq \epsilon$,
	\begin{equation}\label{aib}
	\|T^{-1}(x)\|\leq CB_1^2B_2^2\epsilon^{-1}\tilde{M}(3B_1^2{B_2})^{\tilde{M}}\leq \epsilon^{-2}.
	\end{equation}
	Covering $[-\frac{\delta}{2},\frac{\delta}{2}]^{J}$ by cubes of side  $\mu/4$, and combining (\ref{x0b}) and (\ref{aib}), one has
 \begin{align*}
 	\mathrm{mes}\left\{x\in\left[-\delta/2, \delta/2\right]^{J}:\  \|T^{-1}(x)\|\geq \epsilon^{-2}\right\}&\leq  C\delta^{J}e^{-c\left(\frac{\log \epsilon^{-1}}{\tilde{M}\log(B_1+B_2+B_3)}\right)^{1/J}}\\
 &\leq C\delta^{J}e^{-c\left(\frac{\log \epsilon^{-1}}{{M}\log(B_1+B_2+B_3)}\right)^{1/J}}.
	\end{align*}

\end{proof}

\section*{Acknowledgments}
 I would like to thank  Alan Haynes for telling me  \cite{hugh}, which leads to Lemmas \ref{ledissk} and \ref{ledisskp}. 
 I also wish to thank Svetlana Jitomirskaya for comments on earlier versions of the manuscript. 
  This research was 
 supported by   NSF DMS-1700314/2015683, DMS-2000345 and  the Southeastern
 Conference (SEC) Faculty Travel Grant 2020-2021.

\end{document}